\documentclass[10pt,conference,letterpaper]{IEEEtran}
\IEEEoverridecommandlockouts

\usepackage{times}
\usepackage{latexsym}

\usepackage[T1]{fontenc}
\usepackage[utf8]{inputenc}

\usepackage{microtype}
\usepackage{inconsolata}

\usepackage{graphicx}

\usepackage{mathtools}
\usepackage{amssymb}
\usepackage{amsthm}

\usepackage{siunitx}
\usepackage{threeparttable}
\usepackage{booktabs}
\usepackage{multirow}
\usepackage{tabularx}
\usepackage{array}
\usepackage{xcolor}
\usepackage{colortbl}
\usepackage{placeins}

\usepackage{float}
\usepackage{afterpage}
\usepackage{textcomp}
\usepackage{gensymb}
\usepackage{enumitem}
\usepackage{subcaption}

\usepackage{algorithm}
\usepackage[noend]{algpseudocode}

\newtheoremstyle{boldstyle}
  {\topsep}
  {\topsep}
  {\itshape}
  {}
  {\bfseries}
  {.}
  {.5em}
  {}

\theoremstyle{boldstyle}
\newtheorem{theorem}{Theorem}
\newtheorem{lemma}[theorem]{Lemma}
\newtheorem{proposition}[theorem]{Proposition}

\newtheorem{assumption}{Assumption}
\newtheorem{corollary}{Corollary}

\newcommand{\Coll}{\mathrm{Coll}}
\newcommand{\Bal}{\mathrm{Bal}}

\usepackage{url}    
\newcolumntype{L}[1]{>{\raggedright\arraybackslash}m{#1}}
\newcolumntype{C}[1]{>{\centering\arraybackslash}m{#1}}
\newcolumntype{R}[1]{>{\raggedleft\arraybackslash}m{#1}}

\algrenewcommand\algorithmicrequire{\textbf{Require:}}
\algrenewcommand\algorithmicensure{\textbf{Ensure:}}
\usepackage{hyperref}
\hypersetup{
  colorlinks=true,
  linkcolor=black,
  citecolor=blue,
  urlcolor=blue,
  pdfborder={0 0 0}
}

\def\BibTeX{{\rm B\kern-.05em{\sc i\kern-.025em b}\kern-.08em
  T\kern-.1667em\lower.7ex\hbox{E}\kern-.125emX}}

\title{
Information-Theoretic Limits of Node Localization under Hybrid Graph Positional Encodings
\thanks{
This work was supported by the Major Scientific and Technological Innovation Platform Project of Hunan Province [2024JC1003] and the Graduate Innovation Project of National University of Defense Technology [XJQY2024065]. The authors would like to express their sincere gratitude to all the referees for their careful reading and insightful suggestions.
}
}

\author{%
{ Zimo Yan\textsuperscript{\rm 1}, Zheng Xie\textsuperscript{\rm 1}\thanks{*Corresponding author: Zheng Xie (xiezheng81@nudt.edu.cn)}, Chang Liu\textsuperscript{\rm 1}, Yiqin Lv\textsuperscript{\rm 1}, Runfan Duan\textsuperscript{\rm 1}. }%
\vspace{1.6mm}\\
\fontsize{10}{10}\selectfont\itshape
\textsuperscript{\rm 1}National University of Defense Technology, Changsha, China.
\\\{yanzimo20, xiezheng81, liuchang\_, lvyiqin98, duanrunfan24 \}@nudt.edu.cn}
\fontsize{9}{9}\selectfont\ttfamily\upshape
\fontsize{10}{10}\selectfont\rmfamily\itshape

\fontsize{9}{9}\selectfont\ttfamily\upshape

\begin{document}
\maketitle

\begin{abstract}
Positional encoding has become a standard component in graph learning, especially for graph Transformers and other models that must distinguish structurally similar nodes, yet its fundamental identifiability remains poorly understood. In this work, we study node localization under a hybrid positional encoding that combines anchor-distance profiles with quantized low-frequency spectral features. We cast localization as an observation-map problem whose difficulty is controlled by the number of distinct codes induced by the encoding and establish an information-theoretic converse identifying an impossibility regime jointly governed by the anchor number, spectral dimension, and quantization level. Experiments further support this picture: on random $3$-regular graphs, the empirical crossover is well organized by the predicted scaling, while on two real-world DDI graphs identifiability is strongly graph-dependent, with DrugBank remaining highly redundant under the tested encodings and the Decagon-derived graph becoming nearly injective under sufficiently rich spectral information. Overall, these results suggest that positional encoding should be understood not merely as a heuristic architectural component, but as a graph-dependent structural resolution mechanism.

\vspace{1em}
\noindent\textbf{Keywords:} graph positional encoding; node identifiability; spectral embedding; distance encoding; graph neural networks; drug-drug interaction graphs
\end{abstract}

\section{Introduction}

Graph neural networks and graph Transformers have achieved strong performance on relational data, but their effectiveness often depends on whether node positions can be meaningfully distinguished in the underlying graph. 
This has made positional and structural encodings an important topic in modern graph learning \cite{dwivedi2023benchmark,ying2021graphormer,rampasek2022gps,canturk2024gpse}.

Existing approaches to graph positional encoding broadly fall into three families. 
One family uses \emph{spectral encodings}, typically derived from Laplacian eigenvectors or their invariant variants, to capture global graph geometry \cite{lim2023signnet,huang2024stability}. 
A complementary family uses \emph{distance- or structure-based encodings}, such as shortest-path distances to anchors, random-walk information, or related structural descriptors \cite{li2020distance,gabrielsson2023rewiring}. 
A further line of work focuses on \emph{learned or hybrid encoders}, which combine multiple positional signals within graph Transformers or other graph models \cite{ying2021graphormer,dwivedi2022learnable,canturk2024gpse}.

Despite these advances, three limitations remain. 
Much of the existing literature evaluates positional encoding indirectly through downstream accuracy, leaving open whether the encoding itself can uniquely distinguish nodes. 
Spectral methods also bring sign, basis, and stability issues that require careful treatment \cite{lim2023signnet,huang2024stability}. 
More broadly, the effect of positional encoding is clearly graph-dependent, yet this dependence remains only partially understood from a theoretical perspective \cite{keriven2023role,li2024graphtransformers}.

\textbf{Motivation}:  
This leads to the following question: \textit{Can we characterize when a graph positional encoding contains enough information to localize nodes, and how such identifiability depends jointly on distance features, spectral features, and quantization?}

To address this question, we study node localization through a hybrid observation map that combines anchor-distance profiles with quantized low-frequency spectral features. We then develop an information-theoretic converse analysis for this map and identify an impossibility regime jointly governed by the anchor number, spectral dimension, and quantization level. Experiments on random regular graphs and two real-world DDI graphs further show how the same encoding can behave very differently across graph structures.

\textbf{Outline \& Primary Contributions}:  
The primary contributions of this work are as follows:
\begin{enumerate}
    \item We study node localization through a hybrid observation map that combines anchor-distance features with quantized low-frequency spectral features, and derive a generic information-theoretic converse identifying an impossibility regime governed by the anchor number, spectral dimension, and quantization level.
    \item For the canonical energy embedding, we further sharpen the generic image-size bound through a bucketwise collision-and-balance refinement, yielding a more structured impossibility criterion.
    \item We support the theory with experiments on random $3$-regular graphs and two real-world DDI graph constructions, showing a phase-transition-like crossover in the synthetic setting and a strong graph-dependent contrast between DrugBank and a Decagon-derived graph.
\end{enumerate}
The remainder of the paper is organized as follows. Section~\ref{SE2} reviews related work, Section~\ref{SE3} introduces the problem formulation, Section~\ref{SE5} develops the theory, Sections~\ref{SE7} present the experiments and empirical analysis, and Section~\ref{SE9} concludes the paper.

\section{Literature Review}
\label{SE2}

Graph positional encoding has become a central topic in modern graph learning, especially in graph Transformers and related architectures that require explicit positional or structural information to distinguish structurally similar nodes \cite{zhang2020graphbert,mialon2021graphit,yuan2025gtsurvey}. 
Existing methods can be broadly grouped into three categories: spectral positional encodings, distance- and structure-aware encodings, and learned or hybrid positional/structural encoders.

\subsection{Spectral Positional Encodings}

Spectral positional encodings describe node position through Laplacian eigenvectors, eigenvalues, or related spectral coordinates, and therefore inject global geometric information that is inaccessible to purely local message passing \cite{dwivedi2022learnable,wang2022equivariant,kreuzer2021rethinking}. 
Representative examples include learnable structural and positional representations \cite{dwivedi2022learnable}, equivariant and stable spectral encodings \cite{wang2022equivariant}, and sign/basis-invariant models such as SignNet/BasisNet \cite{lim2023signnet}. 
Recent studies further examine the stability and expressive power of spectral encodings \cite{huang2024stability,zhang2024spectral}, develop alternative spectral or spectral-inspired constructions such as random feature propagation and sheaf-based positional encodings \cite{eliasof2023rfp,he2024sheafpe}, and extend the discussion to directed graphs, where positional information is even harder to define consistently \cite{huang2025directed}. 
At the same time, these methods inherit the ambiguity of eigenvector signs and bases, may react sensitively to graph perturbations, and often rely on eigendecomposition as a nontrivial preprocessing step on large graphs.

\subsection{Distance- and Structure-aware Encodings}

A second line characterizes node position through relative structural information, such as anchor-based distances, shortest-path distances, resistance distances, random-walk statistics, or graph rewiring patterns \cite{you2019pgnn,park2022grpe,black2024comparing}. 
Distance Encoding provides a particularly clear expressivity-oriented formulation of this idea \cite{li2020distance}, while P-GNN and related anchor-based approaches show how positional information can be built directly from distances to sampled reference nodes \cite{you2019pgnn}. 
Graphormer and GRPE incorporate relative structural biases directly into Transformer attention \cite{ying2021graphormer,park2022grpe}, and related work shows that positional information can also be injected through rewiring strategies or analyzed through explicit comparisons between absolute and relative positional schemes in graph Transformers \cite{gabrielsson2023rewiring,black2024comparing}. 
More broadly, structure-enhanced graph encoders also exploit topology-derived similarity information even when the mechanism is not phrased explicitly as positional encoding \cite{zhang2025strucgcn}. 
Recent work has also begun to make the connection between anchor-distance and spectral viewpoints explicit \cite{yan2026bridging}. 
These approaches are often more interpretable than spectral coordinates, but their effectiveness can depend strongly on design choices such as anchor selection, distance definition, and preprocessing, and highly symmetric nodes may still share identical structural signatures.

\subsection{Learned and Hybrid Positional/Structural Encoders}

A third line learns or combines multiple positional signals instead of relying on a single hand-crafted encoding \cite{zhang2020graphbert,yeom2024spegt}. 
Graph-Bert and GraphiT are early examples of Transformer-style graph models that explicitly encode structural and positional information in attention-based architectures \cite{zhang2020graphbert,mialon2021graphit}. 
GPS integrates message passing, Transformer layers, and positional/structural inputs within a unified architecture \cite{rampasek2022gps}, while GRIT emphasizes graph inductive biases without relying on explicit message passing \cite{ma2023grit}. 
GPSE learns transferable positional and structural representations directly from graph structure \cite{canturk2024gpse}, and recent work has further studied foundation-style generalization of positional/structural encodings \cite{franks2025foundation}, efficient positional encoding learning \cite{kanatsoulis2025efficientpe}, large-scale benchmarking across GNNs and graph Transformers \cite{groetschla2024benchmarkpe}, and alternative encoding families such as quantum positional encodings \cite{thabet2024quantum}. 
This line is attractive because of its flexibility, but the added model and training complexity also makes it harder to separate the contribution of the encoding itself from that of the backbone architecture and optimization procedure.

\subsection{Relation to Our Work}

Our work is closest to hybrid approaches that combine global and relative structural information, but differs in objective. 
Most existing studies evaluate positional encoding indirectly through downstream accuracy, transferability, or generalization \cite{franks2025foundation,kanatsoulis2025efficientpe,groetschla2024benchmarkpe}. 
By contrast, we ask whether a fixed encoding can localize nodes at all. 
In this sense, our perspective is more closely related to recent theoretical work on the role of positional encoding in random graphs \cite{keriven2023role}, and to recent analyses that connect anchor-distance encodings with diffusion geometry or study node identifiability under sign- and basis-invariant Laplacian encodings \cite{yan2026bridging,yan2025resolving}. 
However, we focus specifically on node identifiability under a deterministic hybrid observation map that combines anchor distances with quantized low-frequency spectral features. 
This viewpoint is also relevant to graph learning applications such as DDI and polypharmacy modeling \cite{zitnik2018decagon,liu2023m2gcn}, where strong predictive performance does not necessarily imply that the underlying positional encoding provides sufficient structural resolution at the node level.

\section{Problem Formulation}
\label{SE3}

Let $G=(V,E)$ be a simple undirected graph with $|V|=n$. For $u,v\in V$, let
\[
\mathrm{SPD}(u,v)\coloneqq \mathrm{dist}_G(u,v)
\]
denote the shortest-path distance. Given an ordered anchor set
\[
\mathcal A=\{a_1,\dots,a_k\}\subset V,
\]
we define the anchor-distance profile of a vertex $v\in V$ by
\[
\mathbf d_{\mathcal A}(v)
\coloneqq
\bigl(\mathrm{SPD}(v,a_1),\dots,\mathrm{SPD}(v,a_k)\bigr)
\in \mathbb Z_{\ge 0}^k.
\]

Throughout the theory, we consider random $r$-regular graphs
\[
G\sim\mathcal G_{n,r},
\]
where $r\ge 3$ is fixed and $nr$ is even. Let $\mathbf L$ denote a Laplacian-type operator on $G$, such as the combinatorial Laplacian $L=D-A$ or the normalized Laplacian $\mathcal L=I-D^{-1/2}AD^{-1/2}$. Since $G$ is $r$-regular, these differ only by an affine rescaling, so the distinction is immaterial for the formalism below.

Let
\[
0=\lambda_1\le \lambda_2\le \cdots \le \lambda_n
\]
be the eigenvalues of $\mathbf L$, and let $\Phi_m\in\mathbb R^{n\times m}$ be an orthonormal basis for the low-frequency $m$-dimensional spectral subspace under consideration. Because this basis is not unique, we work with a prescribed basis-invariant truncated spectral signature
\[
\mathcal S_m:V\to\mathbb R^{p_m},
\]
which depends only on the underlying spectral subspace and not on the particular choice of eigenbasis. Typical examples include coordinatewise energy signatures or other orthogonally invariant descriptors derived from the row $\Phi_m(v,\cdot)$.

For a quantization resolution $\eta>0$, let
\[
Q_\eta:\mathbb R^{p_m}\to (\eta\mathbb Z)^{p_m}
\]
be a deterministic coordinatewise quantizer. For instance, one may take
\[
(Q_\eta(x))_j
=
\eta\left\lfloor \frac{x_j}{\eta}\right\rfloor,
\quad j=1,\dots,p_m.
\]
The exact form of $Q_\eta$ is not essential here; its role is to model finite resolution and to induce a robust notion of observational indistinguishability.

We combine the distance and spectral components into the observation map
\[
F_{G,\mathcal A}^{(\eta)}:V\to \mathcal Y_{k,m}^{(\eta)}(G,\mathcal A),
\quad
F_{G,\mathcal A}^{(\eta)}(v)
=
\bigl(\mathbf d_{\mathcal A}(v),\,Q_\eta(\mathcal S_m(v))\bigr),
\]
where
\[
\mathcal Y_{k,m}^{(\eta)}(G,\mathcal A)
\coloneqq
\mathrm{Im}\!\bigl(F_{G,\mathcal A}^{(\eta)}\bigr)
\]
is the set of observable codes induced by the encoding.

Now let $v_\star$ be a uniformly random source vertex in $V$, independent of $G$. Conditional on $G$, the anchor set $\mathcal A$ is sampled uniformly without replacement from $V$, independently of $v_\star$. The available data are therefore
\[
F_{G,\mathcal A}^{(\eta)}(v_\star)
=
\bigl(\mathbf d_{\mathcal A}(v_\star),\,Q_\eta(\mathcal S_m(v_\star))\bigr).
\]
The localization problem is to recover $v_\star$ from this observation, given $(G,\mathcal A)$.

For fixed $(G,\mathcal A)$, a reconstruction procedure is any map
\[
s_{G,\mathcal A}:\mathcal Y_{k,m}^{(\eta)}(G,\mathcal A)\to V.
\]
Its conditional error probability is
\[
\mathrm{Err}(s_{G,\mathcal A})
\coloneqq
\mathbb P\!\left(
s_{G,\mathcal A}\bigl(F_{G,\mathcal A}^{(\eta)}(v_\star)\bigr)\neq v_\star
\,\middle|\, G,\mathcal A
\right).
\]
Our goal is to characterize regimes of $(k,m,\eta)$, as functions of $n$, in which localization is information-theoretically impossible, in the sense that
\[
\inf_{s_{G,\mathcal A}}\mathrm{Err}(s_{G,\mathcal A})
\]
remains bounded away from zero with high probability over $(G,\mathcal A)$.

The key combinatorial objects are the preimages of the observation map. For any
\[
y\in \mathcal Y_{k,m}^{(\eta)}(G,\mathcal A),
\]
define
\[
\mathcal P_{G,\mathcal A}^{(\eta)}(y)
\coloneqq
\{v\in V:\,F_{G,\mathcal A}^{(\eta)}(v)=y\}.
\]
When $y=F_{G,\mathcal A}^{(\eta)}(v)$ for some $v\in V$, we also write $\mathcal P_{G,\mathcal A}^{(\eta)}(v)$ for brevity. Large preimages correspond to observational ambiguity: if many vertices induce the same code, then any reconstruction map must assign that code to a single output and hence cannot succeed with probability exceeding the reciprocal of the preimage size to the combinatorial structure of the observation map, especially the size of its image and the geometry of its preimages. This viewpoint underlies the converse analysis in the next section.

\section{Theoretical Analysis}
\label{SE5}

This section develops the theory in two layers. We first establish a generic converse by reducing node localization to an image-size counting problem for the observation map. We then give a model-specific refinement for the canonical energy embedding, where the generic quantized spectral-code complexity bound is sharpened through a bucketwise collision-and-balance principle. All proofs are collected in Appendix~\ref{app:theory-proofs}.

\subsection{Observation-map identity and generic converse}

We begin from a general counting reduction. The key idea is that the optimal localization error is determined exactly by the image size of the observation map, so the problem reduces to bounding how many distinct codes can be induced by the hybrid encoding. The anchor-distance component controls the number of distance buckets, while the spectral component controls how many distinct quantized signatures may appear inside each bucket.

Fix a graph $G=(V,E)$, an anchor set $\mathcal A=\{a_1,\dots,a_k\}\subset V$, a spectral truncation level $m\ge 1$, and a quantization resolution $\eta>0$. Recall the quantized map
\begin{equation}
F_{G,\mathcal A}^{(\eta)}(v)\coloneqq \Big(\mathbf d_{\mathcal A}(v),\ Q_\eta(\mathcal S_m(v))\Big),
\quad v\in V,
\end{equation}
with image
\begin{equation}
\mathcal Y_{k,m}^{(\eta)}(G,\mathcal A)\coloneqq \mathrm{Im}(F_{G,\mathcal A}^{(\eta)}).
\end{equation}
Here $\mathbf d_{\mathcal A}(v)\in\mathbb Z_{\ge 0}^k$ is the anchor distance vector
\begin{equation}
\mathbf d_{\mathcal A}(v)\coloneqq \big(\mathrm{SPD}(v,a_1),\dots,\mathrm{SPD}(v,a_k)\big),
\end{equation}
and $\mathcal S_m:V\to\mathbb R^{p_m}$ is an $O(m)$-invariant spectral embedding, with $Q_\eta$ denoting coordinatewise quantization with bin width $\eta$.

For $y\in\mathcal Y_{k,m}^{(\eta)}(G,\mathcal A)$, define the preimage of $y$ under $F_{G,\mathcal A}^{(\eta)}$ by
\begin{equation}
\mathcal P_{G,\mathcal A}^{(\eta)}(y)\coloneqq \big(F_{G,\mathcal A}^{(\eta)}\big)^{-1}(\{y\})
=
\{v\in V:\ F_{G,\mathcal A}^{(\eta)}(v)=y\}.
\end{equation}
When $y=F_{G,\mathcal A}^{(\eta)}(v)$ we write $\mathcal P_{G,\mathcal A}^{(\eta)}(v)$ for $\mathcal P_{G,\mathcal A}^{(\eta)}(y)$.

In the sequel, the source vertex $v_\star$ is uniformly distributed on $V$ and independent of $(G,\mathcal A)$. A reconstruction map is any measurable function. A measurable section is a special reconstruction map satisfying
\begin{equation}
F_{G,\mathcal A}^{(\eta)}(s_{G,\mathcal A}(y))=y
\end{equation}
for every $y\in \mathcal Y_{k,m}^{(\eta)}(G,\mathcal A)$,
\begin{equation}
s_{G,\mathcal A}:\ \mathcal Y_{k,m}^{(\eta)}(G,\mathcal A)\to V.
\end{equation}
Given $s_{G,\mathcal A}$, its conditional error probability is
\begin{equation}
\mathrm{Err}(s_{G,\mathcal A})\coloneqq \mathbb P\Big(s_{G,\mathcal A}(F_{G,\mathcal A}^{(\eta)}(v_\star))\neq v_\star\ \Big|\ G,\mathcal A\Big).
\end{equation}

\begin{lemma}
\label{lem:preimage-identity}
For any fixed $(G,\mathcal A)$, let
\begin{equation}
Y \coloneqq F_{G,\mathcal A}^{(\eta)}(v_\star).
\end{equation}
Then for any measurable map
\begin{equation}
s_{G,\mathcal A}:\mathcal Y_{k,m}^{(\eta)}(G,\mathcal A)\to V,
\end{equation}
one has
\begin{equation}
\resizebox{\columnwidth}{!}{$
\begin{aligned}
\mathbb P\Big(s_{G,\mathcal A}(Y)=v_\star\ \Big|\ G,\mathcal A\Big)
&=
\mathbb E\Bigg[
\frac{
\mathbf 1\!\left\{s_{G,\mathcal A}(Y)\in \mathcal P_{G,\mathcal A}^{(\eta)}(Y)\right\}
}{
\big|\mathcal P_{G,\mathcal A}^{(\eta)}(Y)\big|
}
\ \Bigg|\ G,\mathcal A
\Bigg]\\
&\le
\mathbb E\Bigg[
\frac{1}{\big|\mathcal P_{G,\mathcal A}^{(\eta)}(v_\star)\big|}
\ \Bigg|\ G,\mathcal A
\Bigg].
\end{aligned}
$}
\end{equation}
Moreover,
\begin{equation}
\mathbb E\Bigg[
\frac{1}{\big|\mathcal P_{G,\mathcal A}^{(\eta)}(v_\star)\big|}
\ \Bigg|\ G,\mathcal A
\Bigg]
=
\frac{\big|\mathrm{Im}(F_{G,\mathcal A}^{(\eta)})\big|}{n}.
\end{equation}
The upper bound is attained by any measurable section
\begin{equation}
\tilde s_{G,\mathcal A}:\mathcal Y_{k,m}^{(\eta)}(G,\mathcal A)\to V
\end{equation}
satisfying
\begin{equation}
F_{G,\mathcal A}^{(\eta)}\big(\tilde s_{G,\mathcal A}(y)\big)=y,
\quad
\forall\, y\in \mathcal Y_{k,m}^{(\eta)}(G,\mathcal A).
\end{equation}
Consequently,
\begin{equation}
\inf_{s_{G,\mathcal A}}\ \mathrm{Err}(s_{G,\mathcal A})
=
1-\frac{\big|\mathrm{Im}(F_{G,\mathcal A}^{(\eta)})\big|}{n}.
\end{equation}
\end{lemma}

For $\mathbf t\in\mathbb Z_{\ge 0}^k$, define the distance bucket
\begin{equation}
B_{\mathbf t}(G,\mathcal A)\coloneqq \{v\in V:\ \mathbf d_{\mathcal A}(v)=\mathbf t\}.
\end{equation}

\begin{lemma}
\label{lem:profile-count}
Assume $\mathrm{diam}(G)\le C_{\mathrm{diam}}\log n$. Then
\begin{equation}
\big|\{\mathbf d_{\mathcal A}(v):v\in V\}\big|\le (C\log n)^k,
\end{equation}
and in particular, the number of nonempty buckets $B_{\mathbf t}(G,\mathcal A)$ is at most $(C\log n)^k$.
\end{lemma}

\begin{assumption}
\label{ass:entropy}
Let $\mathcal S_m(V)\coloneqq \{\mathcal S_m(v):v\in V\}\subset\mathbb R^{p_m}$. Fix $m\ge 1$ and $\eta\in(0,1)$. There exist constants $C_{\mathrm{ent}}\ge 1$, $c_{\mathrm{ent}}>0$ and an event $\mathcal E_{\mathrm{ent}}$ (depending on $G$) with $\mathbb P(\mathcal E_{\mathrm{ent}})=1-o(1)$ such that on $\mathcal E_{\mathrm{ent}}$,
\begin{equation}
\big|Q_\eta(\mathcal S_m(V))\big|
\le
\Big(\frac{C_{\mathrm{ent}}}{\eta}\Big)^{c_{\mathrm{ent}} m}.
\end{equation}
\end{assumption}

Assumption~\ref{ass:entropy} is not intended as the final sharp description of the spectral component. Rather, it serves as a generic complexity hypothesis under which the converse takes a clean information-budget form. The canonical energy embedding is then treated explicitly in Proposition~\ref{prop:entropy-energy} and Corollary~\ref{cor:threshold-energy}.

\begin{lemma}
\label{lem:image-control}
Assume $\mathrm{diam}(G)\le C_{\mathrm{diam}}\log n$ and that Assumption~\ref{ass:entropy} holds on $\mathcal E_{\mathrm{ent}}$. Then on the event $\{\mathrm{diam}(G)\le C_{\mathrm{diam}}\log n\}\cap \mathcal E_{\mathrm{ent}}$,
\begin{equation}
\begin{aligned}
\big|\mathrm{Im}(F_{G,\mathcal A}^{(\eta)})\big|
&\le
(C\log n)^k\cdot \big|Q_\eta(\mathcal S_m(V))\big| \\
&\le
(C\log n)^k\cdot \Big(\frac{C_{\mathrm{ent}}}{\eta}\Big)^{c_{\mathrm{ent}} m},
\end{aligned}
\end{equation}
where $C=C(r)>0$ depends only on $C_{\mathrm{diam}}$.
\end{lemma}

Let $\Phi_m$ be an orthonormal eigenbasis matrix for the first $m$ eigenvalues of $\mathbf L$, and denote by $\Phi_m(v,\cdot)\in\mathbb R^m$ the $v$-th row. Consider the $O(m)$-invariant energy-type spectral embedding
\begin{equation}
\mathcal S_m(v)\coloneqq \big(\Phi_m(v,1)^2,\dots,\Phi_m(v,m)^2\big)\in\mathbb R^m.
\end{equation}

\begin{proposition}
\label{prop:entropy-energy}
For the embedding $\mathcal S_m(v)=(\Phi_m(v,j)^2)_{j\in[m]}$ and $\eta\in(0,1)$, one has
\begin{equation}
\big|Q_\eta(\mathcal S_m(V))\big|
\le
\big|Q_\eta([0,1]^m)\big|
\le
\Big(\frac{2}{\eta}\Big)^m.
\end{equation}
In particular, Assumption~\ref{ass:entropy} holds with $C_{\mathrm{ent}}=2$, $c_{\mathrm{ent}}=1$, and $\mathcal E_{\mathrm{ent}}$ equal to the whole probability space.
\end{proposition}

We can now state the generic converse theorem.

\begin{theorem}
\label{thm:threshold}
Fix $r\ge 3$ and let $G\sim\mathcal G_{n,r}$. Let $v_\star\sim\mathrm{Unif}(V)$ and let $\mathcal A=\{a_1,\dots,a_k\}$ be sampled uniformly without replacement from $V$, independently of $v_\star$ and $G$. Fix $m\ge 1$ and $\eta\in(0,1)$. Assume that the diameter bound $\mathrm{diam}(G)\le C_{\mathrm{diam}}\log n$ holds with probability $1-o(1)$ and that Assumption~\ref{ass:entropy} holds for the spectral embedding $\mathcal S_m$ at resolution $\eta$.
If there exists $\varepsilon_0\in(0,1)$ such that
\begin{equation}
k\log\log n + c_{\mathrm{ent}} m \log\!\Big(\frac{C_{\mathrm{ent}}}{\eta}\Big)\le (1-\varepsilon_0)\log n,
\end{equation}
then there exists $\varepsilon_1\in(0,\varepsilon_0)$ such that with probability $1-o(1)$ over $(G,\mathcal A)$,
\begin{equation}
\inf_{s_{G,\mathcal A}}\ \mathbb P\Big(s_{G,\mathcal A}(F_{G,\mathcal A}^{(\eta)}(v_\star))\neq v_\star\ \Big|\ G,\mathcal A\Big)
\ge
1-n^{-\varepsilon_1}.
\end{equation}
In particular, the left-hand side is bounded below by any fixed constant $\delta\in(0,1)$ for all sufficiently large $n$.
\end{theorem}

Theorem~\ref{thm:threshold} identifies a generic subcritical impossibility regime: when
\begin{equation}
k\log\log n + c_{\mathrm{ent}} m \log\!\Big(\frac{C_{\mathrm{ent}}}{\eta}\Big)
\end{equation}
stays strictly below $\log n$, the encoding cannot produce enough distinct observations to localize a uniformly random node with vanishing error. This is a one-sided converse rather than a sharp threshold theorem.

Assumption~\ref{ass:entropy} should be understood as a generic complexity hypothesis on the quantized spectral codes. Its role is to isolate the information budget governing the converse in an abstract form. For the canonical energy embedding, however, this assumption can be verified explicitly: Proposition~\ref{prop:entropy-energy} yields the concrete constants $C_{\mathrm{ent}}=2$ and $c_{\mathrm{ent}}=1$. Consequently, Theorem~\ref{thm:threshold} specializes to the following assumption-free corollary.

\begin{corollary}
\label{cor:threshold-energy}
Fix $r\ge 3$ and let $G\sim\mathcal G_{n,r}$. Let $v_\star\sim\mathrm{Unif}(V)$ and let $\mathcal A=\{a_1,\dots,a_k\}$ be sampled uniformly without replacement from $V$, independently of $v_\star$ and $G$. Fix $m\ge 1$ and $\eta\in(0,1)$, and take
\begin{equation}
\mathcal S_m(v)=\big(\Phi_m(v,1)^2,\dots,\Phi_m(v,m)^2\big).
\end{equation}
Assume that $\mathrm{diam}(G)\le C_{\mathrm{diam}}\log n$ with probability $1-o(1)$. If there exists $\varepsilon_0\in(0,1)$ such that
\begin{equation}
k\log\log n + m\log\!\Big(\frac{2}{\eta}\Big)\le (1-\varepsilon_0)\log n,
\end{equation}
then there exists $\varepsilon_1\in(0,\varepsilon_0)$ such that with probability $1-o(1)$ over $(G,\mathcal A)$,
\begin{equation}
\inf_{s_{G,\mathcal A}}\ \mathbb P\Big(s_{G,\mathcal A}(F_{G,\mathcal A}^{(\eta)}(v_\star))\neq v_\star\ \Big|\ G,\mathcal A\Big)
\ge
1-n^{-\varepsilon_1}.
\end{equation}
\end{corollary}

Corollary~\ref{cor:threshold-energy} should be viewed as the concrete instantiation of the generic converse theorem for the canonical energy-type spectral embedding.

\subsection{Model-specific refinement via bucketwise collisions and balance}

We now refine the generic counting argument for the canonical energy embedding. Instead of controlling the number of spectral codes in each distance bucket through a global complexity bound, we exploit a more structured principle: if quantized spectral signatures collide sufficiently often inside each non-singleton distance bucket, and if the occupied quantized codes are not excessively unbalanced, then the number of distinct observations within that bucket must be small. This yields a model-specific refinement that is directly connected to the collision and occupancy statistics measured in the experiments.

For any finite $B\subset V$, define the set of quantized spectral codes attained in $B$ by
\begin{equation}
Q_\eta(\mathcal S_m(B))
\coloneqq
\{Q_\eta(\mathcal S_m(v)):\ v\in B\}.
\end{equation}
Let
\begin{equation}
M_{\eta,m}(B)\coloneqq \big|Q_\eta(\mathcal S_m(B))\big|
\end{equation}
denote the number of distinct quantized spectral codes attained in $B$. For each
\begin{equation}
z\in Q_\eta(\mathcal S_m(B)),
\end{equation}
define its occupancy by
\begin{equation}
N_{\eta,m}(B;z)\coloneqq \big|\{v\in B:\ Q_\eta(\mathcal S_m(v))=z\}\big|.
\end{equation}

For any finite $B\subset V$ with $|B|\ge 2$, define the quantized collision density at scale $\eta$ by
\begin{equation}
\resizebox{\columnwidth}{!}{$
\Coll_{\eta,m}^{Q}(B)\coloneqq
\frac{1}{|B|(|B|-1)}\sum_{\substack{u,v\in B\\u\neq v}}
\mathbf 1\Big\{Q_\eta(\mathcal S_m(u))=Q_\eta(\mathcal S_m(v))\Big\}.
$}
\end{equation}
When $|B|=1$, this quantity is not needed, since the bucket already contributes exactly one observation.

We also define the bucket balance coefficient by
\begin{equation}
\Bal_{\eta,m}(B)\coloneqq
\frac{M_{\eta,m}(B)}{|B|}
\cdot
\max_{z\in Q_\eta(\mathcal S_m(B))} N_{\eta,m}(B;z).
\end{equation}
By construction, $\Bal_{\eta,m}(B)\ge 1$. The condition $\Bal_{\eta,m}(B)\le \beta$ means that no occupied quantized code carries more than $\beta$ times the average mass among the occupied codes in $B$.

For the distance partition induced by $\mathcal A$, recall that
\begin{equation}
B_{\mathbf t}(G,\mathcal A)\coloneqq \{v\in V:\ \mathbf d_{\mathcal A}(v)=\mathbf t\},
\end{equation}
and let
\begin{equation}
D(G,\mathcal A)\coloneqq \big|\{\mathbf d_{\mathcal A}(v):v\in V\}\big|
\end{equation}
denote the number of attained distance buckets.

\begin{proposition}
\label{prop:avg-collision-refinement}
Assume $\mathrm{diam}(G)\le C_{\mathrm{diam}}\log n$. Suppose that for every attained distance vector $\mathbf t$, one of the following holds:
\begin{itemize}
    \item $|B_{\mathbf t}(G,\mathcal A)|=1$; or
    \item $|B_{\mathbf t}(G,\mathcal A)|\ge 2$, and
    \begin{equation}
    \Coll_{\eta,m}^{Q}(B_{\mathbf t}(G,\mathcal A))
    \ge
    c_0\,\eta^{c_1 m},
    \end{equation}
    together with
    \begin{equation}
    \Bal_{\eta,m}(B_{\mathbf t}(G,\mathcal A))\le \beta,
    \end{equation}
\end{itemize}
for some constants $c_0,c_1>0$ and $\beta\ge 1$. Then
\begin{equation}
\big|\mathrm{Im}(F_{G,\mathcal A}^{(\eta)})\big|
\le
D(G,\mathcal A)\cdot \Big(1+\frac{\beta}{c_0\,\eta^{c_1 m}}\Big).
\end{equation}
In particular,
\begin{equation}
\big|\mathrm{Im}(F_{G,\mathcal A}^{(\eta)})\big|
\le
(C\log n)^k\cdot \Big(1+\frac{\beta}{c_0\,\eta^{c_1 m}}\Big).
\end{equation}
\end{proposition}

Proposition~\ref{prop:avg-collision-refinement} should be understood as a second-layer counting principle. Singleton buckets are already resolved by the distance component alone, while each non-singleton bucket is further controlled by two measurable quantities: a lower bound on within-bucket quantized collisions and an upper bound on the imbalance of occupied quantized codes.

\begin{corollary}
\label{cor:image-control-refined}
Assume $\mathrm{diam}(G)\le C_{\mathrm{diam}}\log n$ with probability $1-o(1)$, and suppose there exists an event $\mathcal E_{\mathrm{ref}}$ with $\mathbb P(\mathcal E_{\mathrm{ref}})=1-o(1)$ such that on $\mathcal E_{\mathrm{ref}}$, for every attained distance vector $\mathbf t$, one of the following holds:
\begin{itemize}
    \item $|B_{\mathbf t}(G,\mathcal A)|=1$; or
    \item $|B_{\mathbf t}(G,\mathcal A)|\ge 2$, and
    \begin{equation}
    \Coll_{\eta,m}^{Q}(B_{\mathbf t}(G,\mathcal A))
    \ge
    c_0\,\eta^{c_1 m},
    \end{equation}
    together with
    \begin{equation}
    \Bal_{\eta,m}(B_{\mathbf t}(G,\mathcal A))\le \beta,
    \end{equation}
\end{itemize}
for some constants $c_0,c_1>0$ and $\beta\ge 1$. Then with probability $1-o(1)$ over $(G,\mathcal A)$,
\begin{equation}
\big|\mathrm{Im}(F_{G,\mathcal A}^{(\eta)})\big|
\le
(C\log n)^k\cdot \Big(1+\frac{\beta}{c_0\,\eta^{c_1 m}}\Big).
\end{equation}
In particular, since $\eta\in(0,1)$,
\begin{equation}
\big|\mathrm{Im}(F_{G,\mathcal A}^{(\eta)})\big|
\le
(C\log n)^k\cdot \Big(1+\frac{\beta}{c_0}\Big)\eta^{-c_1 m}.
\end{equation}
\end{corollary}

Combining the refined image-size bound with Lemma~\ref{lem:preimage-identity} yields the following refined impossibility corollary.

\begin{corollary}
\label{cor:threshold-refined}
In the setting of Theorem~\ref{thm:threshold}, further assume that there exist constants $c_0,c_1>0$ and $\beta\ge 1$, and an event $\mathcal E_{\mathrm{ref}}$ with $\mathbb P(\mathcal E_{\mathrm{ref}})=1-o(1)$ such that on $\mathcal E_{\mathrm{ref}}$, for every attained distance vector $\mathbf t$, one of the following holds:
\begin{itemize}
    \item $|B_{\mathbf t}(G,\mathcal A)|=1$; or
    \item $|B_{\mathbf t}(G,\mathcal A)|\ge 2$, and
    \begin{equation}
    \Coll_{\eta,m}^{Q}(B_{\mathbf t}(G,\mathcal A))
    \ge
    c_0\,\eta^{c_1 m},
    \end{equation}
    together with
    \begin{equation}
    \Bal_{\eta,m}(B_{\mathbf t}(G,\mathcal A))\le \beta.
    \end{equation}
\end{itemize}
If there exists $\varepsilon_0\in(0,1)$ such that
\begin{equation}
k\log\log n + c_1 m \log(1/\eta)\le (1-\varepsilon_0)\log n,
\end{equation}
then there exists $\varepsilon_1\in(0,\varepsilon_0)$ such that with probability $1-o(1)$ over $(G,\mathcal A)$,
\begin{equation}
\inf_{s_{G,\mathcal A}}\ \mathbb P\Big(s_{G,\mathcal A}(F_{G,\mathcal A}^{(\eta)}(v_\star))\neq v_\star\ \Big|\ G,\mathcal A\Big)
\ge
1-n^{-\varepsilon_1}.
\end{equation}
\end{corollary}

Corollary~\ref{cor:threshold-refined} is a model-specific refinement of Theorem~\ref{thm:threshold}. Once bucketwise quantized collisions are sufficiently strong and the occupied quantized codes remain moderately balanced inside each non-singleton distance bucket, the generic quantized spectral-code complexity bound can be replaced by a sharper bucketwise counting principle.

\section{Empirical Evaluation}
\label{SE7}

\subsection{Experimental setup}
We evaluate the proposed encoding on both synthetic and real-world graphs, focusing on the identifiability of the encoding itself rather than on end-to-end prediction. The object under study is the deterministic observation map
\begin{equation}\label{eq43}
\begin{aligned}
F(v)=\bigl(d_A(v),Q_{\eta}(S_m(v))\bigr),
\end{aligned}
\end{equation}
where \(d_A(v)\) is the anchor-distance profile and \(Q_{\eta}(S_m(v))\) is the relatively quantized low-frequency spectral code. Unless otherwise stated, anchors are sampled uniformly without replacement, the spectral component uses the first \(m\) non-trivial eigenvectors of the normalized Laplacian, and quantization is relative.

For the synthetic study, we use random \(3\)-regular graphs with
\begin{equation*}
\begin{aligned}
n &\in \{500,1000,2000,4000\},\\
k &\in \{1,2,3,4,6,8\},\\
m &\in \{0,1,2,5\},\\
\eta &\in \{0.9,0.7,0.5,0.3,0.1\},
\end{aligned}
\end{equation*}
averaging over \(20\) independently sampled graphs for each configuration; the case \(m=0\) is included as a distance-only baseline. We organize the synthetic results by
\[
\rho_{\mathrm{eng}}
=
\frac{k\log\log n + m\log(2/\eta)}{\log n},
\]
and also report the empirical anchor threshold
\[
k_{\mathrm{emp}}(n,m,\eta)
=
\min
\left\{
k:
\mathbb{E}[\mathrm{err}(F)]\le 0.1
\right\}.
\]

For the real-world study, we consider two DDI graph constructions: a DrugBank graph with
\[
(N_{\mathrm{DB}},E_{\mathrm{DB}})=(1684,189774),
\]
and a Decagon-derived graph with
\[
(N_{\mathrm{DEC}},E_{\mathrm{DEC}})=(40,466).
\]
On these fixed graphs, we evaluate
\[
k\in\{1,2,3,4,6,8\},\quad
m\in\{0,2,5,10\},\quad
\eta\in\{0.9,0.5,0.25,0.1\},
\]
and average results over \(30\) independent random anchor samplings.

Our primary metric is the normalized image size
\begin{equation}\label{eq50}
\begin{aligned}
\mathrm{succ}(F)
&=
\frac{|\mathrm{Im}(F)|}{n},
\\
\mathrm{err}(F)
&=
1-\mathrm{succ}(F).
\end{aligned}
\end{equation}
A smaller error therefore indicates stronger node identifiability. Full dataset construction, baseline definitions, quantization rules, and detailed diagnostic protocols are deferred to Appendix~\ref{app:exp-details}.

\subsection{Phase-transition-like behavior on random regular graphs}

Figure~\ref{fig:expA_main} shows the phase-transition-like behavior most clearly in the random-regular setting. 
In panel~(a), the mean localization error drops sharply once the anchor budget enters the transition region.

Panel~(b) shows the corresponding structural effect: when plotted against the theory-guided budget ratio \(\rho_{\mathrm{eng}}\), the mean average preimage size rapidly collapses toward one. This indicates that the empirical error drop is accompanied by a genuine shrinkage of the fibers of the observation map.

The right-hand block in panel~(c) further shows that, across spectral dimensions \(m\in\{0,1,2,5\}\), plotting the error against \(\rho_{\mathrm{eng}}\) organizes the transition more consistently across graph sizes than the anchor number alone. The main crossover remains concentrated near \(\rho_{\mathrm{eng}}\approx 1\), although the sharpness of the transition varies with \(m\). This should be interpreted as an empirical scaling law aligned with the converse boundary, rather than as evidence of a sharp threshold.

Taken together, these plots show that the transition is governed more naturally by the theory-guided information budget than by the raw anchor count alone, and that the error drop coincides with a combinatorial collapse of observational ambiguity.

\begin{figure*}[htbp]
    \centering
    \includegraphics[width=0.98\textwidth]{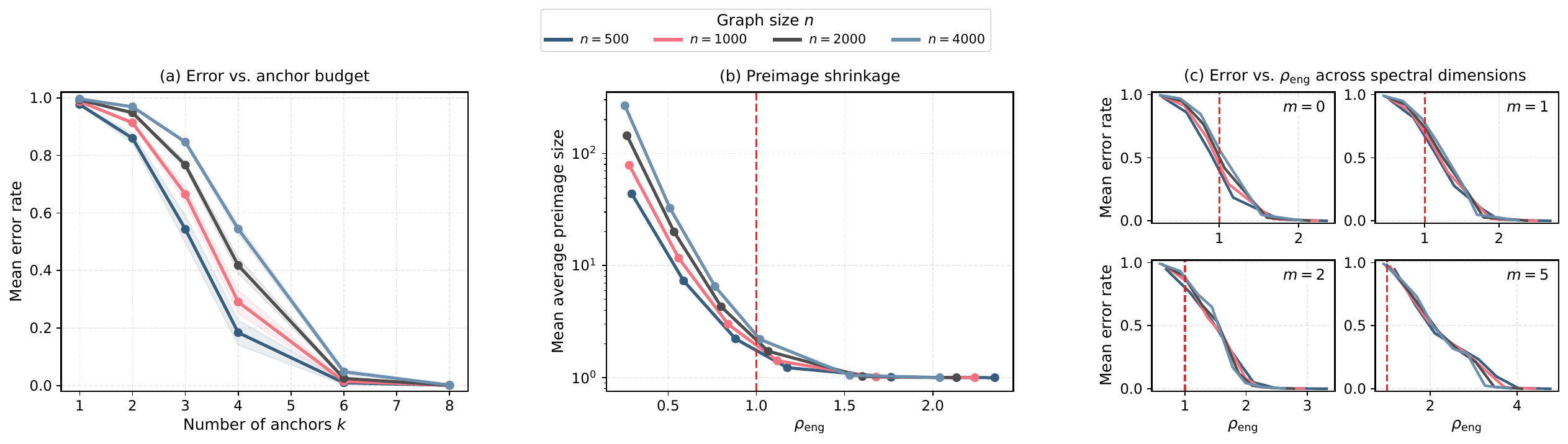}
    \caption{Phase-transition-like behavior on random \(3\)-regular graphs. 
    (a) Mean localization error versus the number of anchors \(k\). 
    (b) Mean average preimage size versus the budget ratio \(\rho_{\mathrm{eng}}\), showing the collapse of observation-map fibers in the low-error regime. 
    (c) Error curves plotted against \(\rho_{\mathrm{eng}}\) for spectral dimensions \(m\in\{0,1,2,5\}\). Across spectral dimensions, the transition is more consistently organized by the theory-guided budget ratio than by the anchor number alone.}
    \label{fig:expA_main}
\end{figure*}

\subsection{Theory-aligned diagnostics of the observation map}
\label{subsec:theory_aligned_diagnostics}
Figure~\ref{fig:theory_diag_kemp} reports \(k_{\mathrm{emp}}\) across \(n\), \(m\), and \(\eta\).

The dominant pattern is the stability of the distance-only baseline: \(k_{\mathrm{emp}}=6\) for all tested \(n\) and \(\eta\) when \(m=0\). Once the spectral channel is added, the threshold decreases, and the gain is strongest at finer quantization. At \(\eta=0.1\), the threshold drops from \(6\) to \((3,4,4,6)\) for \(m=1\), to \((2,2,3,3)\) for \(m=2\), and to \(1\) for all tested graph sizes when \(m=5\). The curves are not strictly monotone in \(\eta\), which is expected since \(k_{\mathrm{emp}}\) is an integer threshold estimated from finitely many trials.

Table~\ref{tab:theory_diag_eta01} shows the corresponding threshold statistics at \(\eta=0.1\). At threshold, the normalized image size is already high (\(0.9057\)–\(0.9999\)), and the average preimage size stays close to one. The spectral code count increases sharply with \(m\), reaching nearly \(n\) for \(m=5\); this is precisely the regime where the anchor threshold collapses to \(1\).

\begin{figure}[t]
    \centering
    \includegraphics[width=0.85\linewidth]{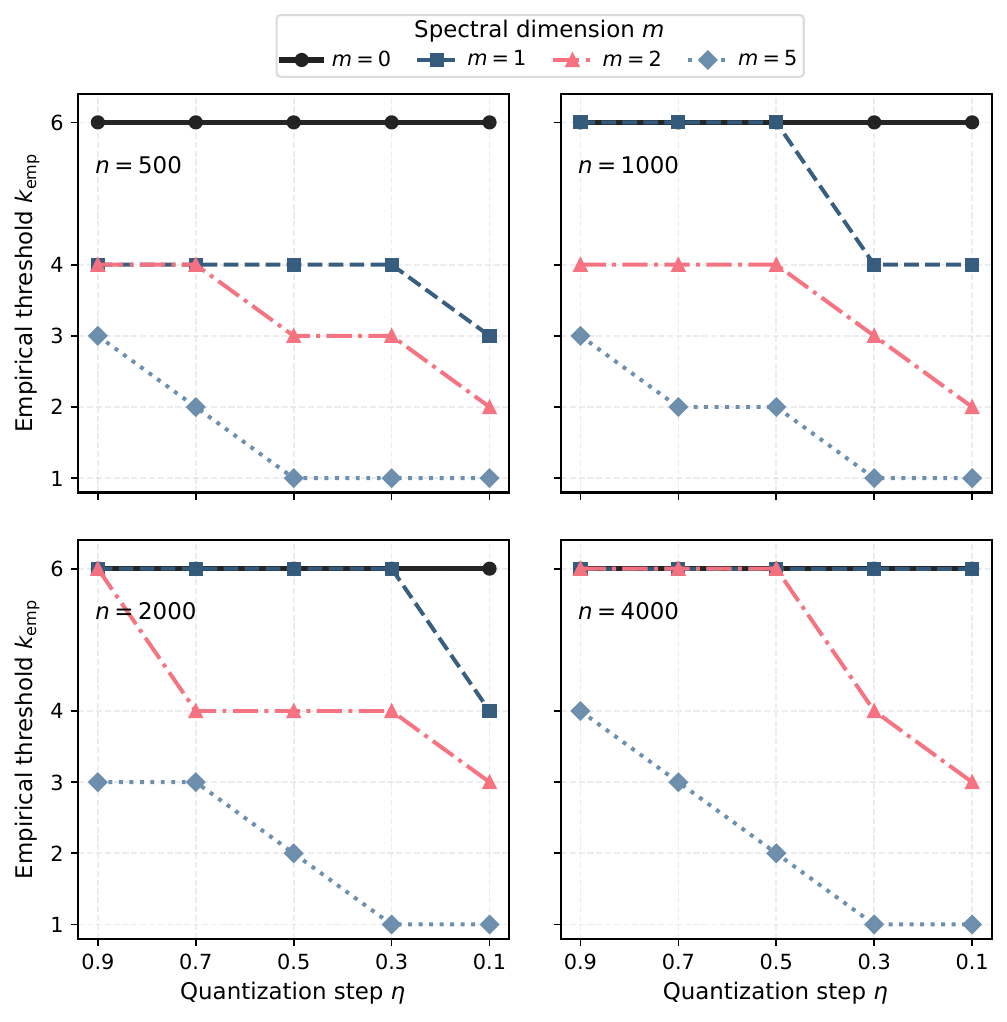}
    \caption{Empirical threshold \(k_{\mathrm{emp}}\) versus quantization step \(\eta\) for random \(3\)-regular graphs. Each panel fixes \(n\), and each curve corresponds to a spectral dimension \(m\).}
    \label{fig:theory_diag_kemp}
\end{figure}

\begin{table}[t]
\centering
\small
\caption{Threshold diagnostics at the finest quantization level $\eta=0.1$.}
\label{tab:theory_diag_eta01}

\resizebox{\linewidth}{!}{%
\begin{tabular}{ccrrrrr}
\toprule
\toprule
$n$ & $m$ & $k_{\mathrm{emp}}$ & $\rho_{\mathrm{emp}}$ & $|\mathrm{Im}(F)|/n$ & Avg.\ preimage & $|Q_\eta(\widetilde S_m(V))|$ \\
\midrule
\multirow{4}{*}{500}
  & 0 & 6 & 1.764 & 0.992 & 1.008 & 1 \\
  & 1 & 3 & 1.364 & 0.906 & 1.105 & 51.75 \\
  & 2 & 2 & 1.552 & 0.956 & 1.046 & 285.6 \\
  & 5 & 1 & 2.704 & 1.000 & 1.000 & 499.45 \\
\specialrule{0.08em}{0.15em}{0.15em}

\multirow{4}{*}{1000}
  & 0 & 6 & 1.679 & 0.987 & 1.014 & 1 \\
  & 1 & 4 & 1.553 & 0.963 & 1.038 & 66 \\
  & 2 & 2 & 1.427 & 0.915 & 1.093 & 444.2 \\
  & 5 & 1 & 2.448 & 0.999 & 1.001 & 997.8 \\
\specialrule{0.08em}{0.15em}{0.15em}

\multirow{4}{*}{2000}
  & 0 & 6 & 1.601 & 0.973 & 1.027 & 1 \\
  & 1 & 4 & 1.462 & 0.937 & 1.067 & 76.3 \\
  & 2 & 3 & 1.589 & 0.967 & 1.034 & 692.75 \\
  & 5 & 1 & 2.237 & 0.999 & 1.001 & 1990.1 \\
\specialrule{0.08em}{0.15em}{0.15em}

\multirow{4}{*}{4000}
  & 0 & 6 & 1.530 & 0.948 & 1.055 & 1 \\
  & 1 & 6 & 1.892 & 0.995 & 1.005 & 87.65 \\
  & 2 & 3 & 1.488 & 0.940 & 1.064 & 1033.6 \\
  & 5 & 1 & 2.061 & 0.998 & 1.002 & 3962 \\
\bottomrule
\bottomrule
\end{tabular}%
}

\vspace{2pt}
\parbox{\linewidth}{\footnotesize
\textit{Note.} Larger spectral codebooks are associated with smaller anchor thresholds. For $m=5$, the codebook size is already close to $n$, and $k_{\mathrm{emp}}=1$ for all tested graph sizes.
}
\end{table}

\subsection{Bucketwise refinement diagnostics}

To examine the mechanism behind the refined bucketwise counting bound, we analyze spectral refinement inside non-singleton distance buckets on random $3$-regular graphs with $n=2000$, using three representative regimes:
collision-heavy $(k,m,\eta)=(2,1,2.0)$,
transition $(2,2,1.0)$,
and near-injective $(2,5,0.3)$.
These three regimes are auxiliary diagnostic settings chosen to span a collision-heavy, transition, and near-injective spectrum, and are not restricted to the main grid used in the phase-transition experiments. For each regime, we use 20 graph instances and 10 anchor resamples per graph. Since the same graph instances and the same anchor resamples are used across the three regimes, the distance buckets are identical across rows, and the comparison isolates the effect of the spectral component within the residual buckets.

Table~\ref{tab:bucketwise_refinement} shows a clear progression. In the collision-heavy regime, the pair-weighted within-bucket collision is large ($0.727$), while the median normalized within-bucket code count is small ($0.185$), indicating substantial compression inside non-singleton buckets. In the near-injective regime, the weighted collision is negligible ($4.45\times 10^{-4}$) and the median normalized code count is essentially one, showing that the spectral component almost completely resolves the remaining buckets. The transition regime lies between these two extremes.

The balance statistic further sharpens this picture. The transition regime has the largest upper-tail balance coefficient, indicating that collision alone does not determine the within-bucket image size: more uneven occupancy is associated with a larger number of occupied spectral codes. This is consistent with the refined proposition, where both collision and balance enter the bucketwise counting bound. The same ordering remains stable when restricting attention to larger non-singleton buckets, and becomes even more pronounced for the normalized within-bucket code count when the cutoff is increased from $|B_{\mathbf t}|\ge 3$ to $|B_{\mathbf t}|\ge 10$.

\begin{table*}[t]
  \centering
  \scriptsize
  \setlength{\tabcolsep}{5pt}
  \caption{Bucketwise refinement diagnostics on random $3$-regular graphs with $n=2000$. Values are averaged over 20 graph trials and 10 anchor resamples per trial. The last four columns report large-bucket diagnostics under the cutoffs $|B_{\mathbf t}|\ge 3$ and $|B_{\mathbf t}|\ge 10$.}
  \label{tab:bucketwise_refinement}
  \resizebox{\textwidth}{!}{%
  \begin{tabular}{lcccccccc}
    \toprule
    \textbf{Regime} 
    & \textbf{Overall Error} 
    & \textbf{Singleton Frac.} 
    & \textbf{Weighted Collision} 
    & \textbf{Median $M(B_{\mathbf t})/|B_{\mathbf t}|$} 
    & \textbf{$q_{0.9}$ Balance} 
    & \textbf{Large Coll. ($\ge 3$)} 
    & \textbf{Large Coll. ($\ge 10$)} 
    & \textbf{Large Median $M/|B|$ ($\ge 3,\ge 10$)} \\
    \midrule
    Collision-heavy 
    & $0.894$ & $0.202$ & $\mathbf{0.727}$ & $\mathbf{0.185}$ & $3.23$ 
    & $0.727$ & $0.727$ & $0.154,\ 0.093$ \\

    Transition      
    & $0.703$ & $0.202$ & $0.253$ & $0.465$ & $\mathbf{6.97}$ 
    & $0.253$ & $0.253$ & $0.413,\ 0.311$ \\

    Near-injective  
    & $\mathbf{0.012}$ & $0.202$ & $4.45\times 10^{-4}$ & $0.9999$ & $1.89$ 
    & $4.45\times 10^{-4}$ & $4.46\times 10^{-4}$ & $0.9999,\ 0.9999$ \\
    \bottomrule
  \end{tabular}%
  }
\end{table*}

\subsection{Identifiability on two real-world DDI graphs}
\label{subsec:real_ident}

This experiment is intended as an intra-domain case study rather than a claim of universality across graph families. We compare identifiability on the two real-world DDI graphs introduced in the setup.

The two graphs show sharply different identifiability. On DrugBank, the observation map remains highly colliding for all tested settings. Even at $(k,m,\eta)=(8,10,0.1)$, we have
\[
\mathrm{err}(F)=0.9023,\quad \frac{|\mathrm{Im}(F)|}{n}=0.0977,
\]
with vertex-mean preimage size $79.33$ and singleton-preimage vertex fraction $0.0384$.
\begin{table}[t]
\centering
\caption{Structural statistics of the two real-world DDI graphs.}
\label{tab:real_graph_stats}
\begin{tabular}{lcc}
\toprule
Statistic & DrugBank DDI & Decagon-derived DDI \\
\midrule
$n$ & 1684 & 40 \\
$|E|$ & 189774 & 466 \\
Average degree & 225.385 & 23.300 \\
Density & 0.1339 & 0.5974 \\
Diameter & 5 & 3 \\
Avg. shortest-path length & 2.086 & 1.414 \\
Avg. clustering coefficient & 0.550 & 0.874 \\
Transitivity & 0.496 & 0.801 \\
Degree variance & 38041.436 & 106.610 \\
Degree Gini & 0.478 & 0.246 \\
\bottomrule
\end{tabular}
\end{table}

On the small Decagon-derived graph considered here, distance information alone is insufficient, but adding spectral information makes the encoding nearly injective under the tested settings. At $(8,0,0.1)$, the error is $0.4250$ and the normalized image size is $0.5750$. At $(8,10,0.1)$, these become
\[
\mathrm{err}(F)=0.0158,\quad \frac{|\mathrm{Im}(F)|}{n}=0.9842,
\]
with vertex-mean preimage size $1.03$ and singleton-preimage vertex fraction $0.9683$.

The aggregate trends are consistent with this contrast. On the Decagon-derived graph, the mean error decreases from $0.704$ to $0.192$ as $m$ increases from $0$ to $10$, from $0.579$ to $0.293$ as $\eta$ decreases from $0.9$ to $0.1$, and from $0.585$ to $0.264$ as $k$ increases from $1$ to $8$. On DrugBank, the corresponding changes are much smaller.

Table~\ref{tab:real_graph_compare} shows the same pattern at the bucket level. On DrugBank, singleton distance buckets are rare and non-singleton buckets remain strongly colliding. On the Decagon-derived graph, the near-injective regime is associated with more singleton distance buckets and much weaker within-bucket collisions. Across these two DDI graphs, identifiability depends not only on \((k,m,\eta)\), but also on the observation-map geometry induced by the underlying graph structure.

\begin{table*}[t]
\centering
\caption{Representative identifiability results on the two real-world DDI graphs. Here $|\mathrm{Im}(F)|/n$ is the normalized image size, and the remaining columns report preimage- and distance-bucket-based diagnostics.}
\label{tab:real_graph_compare}
\resizebox{\textwidth}{!}{%
\begin{tabular}{ccccccccccc}
\toprule
Graph & $k$ & $m$ & $\eta$ & Error & $|\mathrm{Im}(F)|/n$ & Vertex-mean preimage & Singleton-preimage frac. & Singleton distance-bucket frac. & Median non-singleton collision & Median non-singleton balance \\
\midrule
DrugBank DDI & 1 & 0  & 0.5 & 0.9971 & 0.0029 & 959.65 & 0.0007 & 0.0007 & 1.0000 & 1.0000 \\
DrugBank DDI & 1 & 10 & 0.1 & 0.9918 & 0.0082 & 949.07 & 0.0037 & 0.0007 & 0.9637 & 3.4002 \\
DrugBank DDI & 4 & 0  & 0.1 & 0.9773 & 0.0227 & 264.77 & 0.0057 & 0.0057 & 1.0000 & 1.0000 \\
DrugBank DDI & 4 & 10 & 0.1 & 0.9711 & 0.0289 & 288.81 & 0.0104 & 0.0064 & 1.0000 & 1.0000 \\
DrugBank DDI & 8 & 0  & 0.1 & 0.9093 & 0.0907 & 92.46 & 0.0340 & 0.0340 & 1.0000 & 1.0000 \\
DrugBank DDI & 8 & 10 & 0.1 & 0.9023 & 0.0977 & 79.33 & 0.0384 & 0.0346 & 1.0000 & 1.0000 \\
\midrule
Decagon-derived DDI & 1 & 0  & 0.5 & 0.9175 & 0.0825 & 24.14 & 0.0283 & 0.0283 & 1.0000 & 1.0000 \\
Decagon-derived DDI & 1 & 10 & 0.1 & 0.0233 & 0.9767 & 1.05 & 0.9533 & 0.0275 & 0.0015 & 1.3520 \\
Decagon-derived DDI & 4 & 0  & 0.1 & 0.6958 & 0.3042 & 9.47 & 0.1792 & 0.1792 & 1.0000 & 1.0000 \\
Decagon-derived DDI & 4 & 10 & 0.1 & 0.0200 & 0.9800 & 1.04 & 0.9600 & 0.1792 & 0.0000 & 1.0000 \\
Decagon-derived DDI & 8 & 0  & 0.1 & 0.4250 & 0.5750 & 3.83 & 0.4367 & 0.4367 & 1.0000 & 1.0000 \\
Decagon-derived DDI & 8 & 10 & 0.1 & 0.0158 & 0.9842 & 1.03 & 0.9683 & 0.4458 & 0.0000 & 1.0000 \\
\bottomrule
\end{tabular}%
}
\end{table*}

\subsection{Ablation studies}
\label{subsec:ablation}

We conduct ablations on random $3$-regular graphs to assess the roles of the distance and spectral components, the interaction of $(k,m,\eta)$, and the sensitivity to anchor selection. Unless otherwise stated, the ablations use
\begin{equation*}
\begin{aligned}
n &\in \{500,1000,2000\},\\
k &\in \{1,2,3,4,6,8\},\\
m &\in \{0,2,5,10\},\\
\eta &\in \{0.9,0.5,0.25,0.1\},
\end{aligned}
\end{equation*}
averaged over \(20\) independent graph trials.

\paragraph{Representation and anchors.}
Table~\ref{tab:ablation_repr_anchor} shows that positional information is necessary. Without it, the error is nearly one. Spectral-only features remain weak, with mean error $0.9328$. Distance-only features are much more informative, reducing the mean error to $0.2912$. The full observation map further reduces the mean error to $0.0789$, with mean preimage size $1.0908$ and singleton-observation ratio $0.9310$. On random regular graphs, anchor distances provide the main discriminative signal, while the quantized spectral term mainly refines nodes within the same distance bucket.

More sophisticated anchor strategies yield only limited improvements over uniform random sampling. The mean error is $0.0941$ for random anchors, $0.0832$ for the farthest-point heuristic, and $0.0755$ for the degree-based heuristic. Random anchors are also stable under resampling: the mean within-graph standard deviation of the error is $0.0150$ for $n=1000$ and $0.0199$ for $n=2000$, with maximum observed value $0.0290$.

\paragraph{Hyperparameter interaction.}
Figure~\ref{fig:ablation_interaction} shows that the effects of $k$, $m$, and $\eta$ are coupled. For fixed $k=4$ and $n=2000$, the mean error decreases from $0.4156$ at $m=0$ to $0.1308$ at $(m,\eta)=(2,0.1)$, $0.0393$ at $(5,0.1)$, and $9.5\times 10^{-4}$ at $(10,0.1)$. For fixed $m=5$ and $n=2000$, the mean error decreases from $0.5501$ at $(k,\eta)=(1,0.1)$ to $0.0314$ at $(4,0.1)$, $5.75\times 10^{-4}$ at $(6,0.1)$, and $7.5\times 10^{-5}$ at $(8,0.1)$. Smaller $\eta$ provides an additional gain at fixed $k$. These results are consistent with the theoretical picture that anchor distances determine the coarse partition and spectral information improves resolution within distance buckets.

Further results on quantization calibration and additional diagnostic details are given in Appendix~\ref{app:exp-details}.

\begin{table}[t]
\centering
\scriptsize
\setlength{\tabcolsep}{3pt}
\caption{Representation ablation and anchor strategy on random $3$-regular graphs.}
\label{tab:ablation_repr_anchor}
\begin{tabularx}{\columnwidth}{@{}l>{\raggedright\arraybackslash}Xcccc@{}}
\toprule
Setting & Description & Error & Preimg. & Single. & Coll. \\
\midrule
\multicolumn{6}{c}{\textbf{Representation ablation}} \\
\midrule
NoPE     & no positional encoding & 0.999 & 1166.7 & 0.000 & 1.000 \\
Spectral & spectral only          & 0.933 & 25.00  & 0.261 & 0.198 \\
Distance & anchor distances only  & 0.291 & 1.439  & 0.743 & 0.0009 \\
Full     & distance + spectral    & \textbf{0.079} & \textbf{1.091} & \textbf{0.931} & \textbf{0.0002} \\
\midrule
\multicolumn{6}{c}{\textbf{Anchor strategy (full feature)}} \\
\midrule
Random   & uniform anchors         & 0.094 & 1.112 & 0.920 & 0.0002 \\
Farthest & farthest-point heuristic& 0.083 & 1.097 & 0.926 & 0.0002 \\
Degree   & degree-based heuristic  & \textbf{0.076} & \textbf{1.087} & \textbf{0.934} & \textbf{0.0002} \\
\bottomrule
\end{tabularx}
\end{table}

\begin{figure}[htbp]
    \centering
    \includegraphics[width=0.95\linewidth]{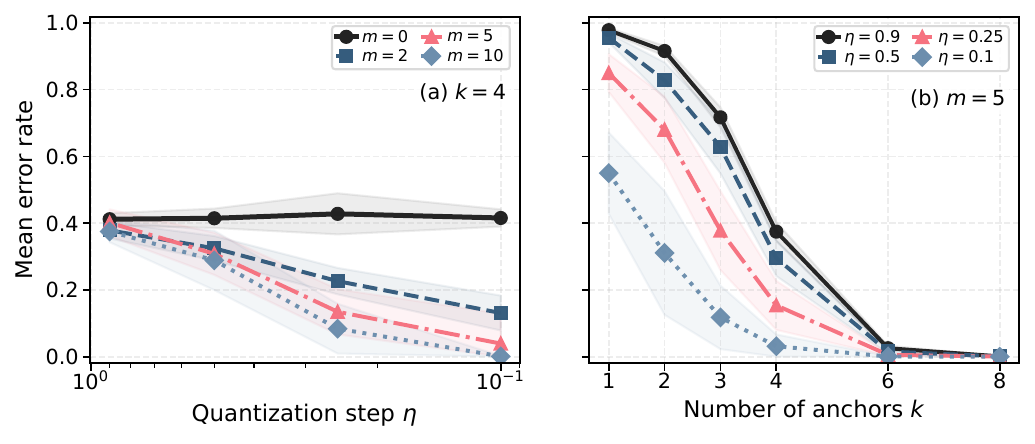}
    \caption{Hyperparameter interaction on random $3$-regular graphs at $n=2000$. 
    Top: mean error versus $\eta$ for fixed $k=4$. 
    Bottom: mean error versus $k$ for fixed $m=5$.}
    \label{fig:ablation_interaction}
\end{figure}

\section{Discussion and Conclusions}
\label{SE9}

\subsection{Discussion}

This work studies node localization under a hybrid graph positional encoding that combines anchor-distance profiles with quantized low-frequency spectral features. 
The main theoretical message is that localization is governed by an information budget jointly determined by the context size $k$, the spectral dimension $m$, and the quantization resolution $\eta$. 
Our converse theory formalizes this through the image size of the observation map $F$, and shows that when the induced code space is too small relative to $n$, the encoding is information-theoretically insufficient to localize a uniformly random node with vanishing error.

The experiments support this picture from two complementary perspectives. 
On random regular graphs, the empirical phase-transition-like crossover is well organized by the theory-guided ratio $\rho_{\mathrm{eng}}(n,k,m,\eta)$, which is consistent with the scaling suggested by the converse theory. 
On the two DDI graphs considered here, identifiability is clearly graph-dependent: DrugBank remains highly colliding under all tested settings, whereas the small Decagon-derived graph becomes nearly injective once sufficiently rich and finely quantized spectral information is included. 
Overall, the proposed encoding can be viewed as a graph-dependent structural resolution mechanism whose effectiveness depends jointly on $(k,m,\eta)$ and on the redundancy level of the underlying graph.

\subsection{Limitations}

Several limitations should be noted. 
First, the present theory is mainly one-sided: it establishes impossibility in a subcritical regime, but does not yet provide matching achievability guarantees or a sharp transition boundary. 
Accordingly, the empirical crossover should be interpreted as a phase-transition-like scaling law aligned with the converse boundary, rather than as a proof of a sharp threshold.

Second, the current analysis is developed for random regular graphs and for a specific orthogonally invariant spectral summary. 
Its present scope therefore does not yet cover broader graph models, more heterogeneous structural regimes, or learned positional representations.

Third, the real-world validation is intentionally narrow. 
We study two representative DDI graph constructions to support an intra-domain graph-dependence claim, but not to claim universality across substantially different graph families.

Fourth, the main-text experiments use a calibrated relative quantization rule, whereas the core theory is formulated in terms of a fixed absolute-resolution quantizer. 
Although the appendix shows that the qualitative trend persists under properly calibrated absolute quantization, a tighter theory-experiment match is still desirable.

Finally, the experiments focus on the identifiability of the encoding itself rather than on a full end-to-end downstream predictor. 
The present results should therefore be interpreted as structural limits and diagnostics for positional encoding, not as a complete account of predictive performance in learned graph models.

\subsection{Future Directions}

A natural next step is to develop matching converse-achievability results, so that the transition can be characterized more sharply and under weaker assumptions. 
It would also be valuable to extend the framework beyond random regular graphs to broader graph families, including weighted, heterogeneous, and dynamic graphs.

Another important direction is to tighten the match between the abstract absolute-resolution quantizer used in the theory and the calibrated quantization rules used in practice. 
On the methodological side, it would also be interesting to replace the current handcrafted encoding by adaptive or learnable variants, such as learnable anchor selection, multi-scale spectral summaries, and task-aware quantization rules.

On the empirical side, a key extension is to connect identifiability more directly to downstream graph learning performance, so that the present analysis can serve not only as a structural criterion, but also as a practical principle for designing positional encodings in graph neural networks and graph Transformers.

\section*{Declarations}

\paragraph{Availability of data and materials}
The source code for the MetaMolGen model developed in this study is publicly available on GitHub at \url{https://github.com/yzz980314/Graph_Node_Identifiability_via_Observation_Maps}. 
The datasets used in this study are publicly available. The DrugBank database can be accessed from \url{https://go.drugbank.com/}. The Decagon polypharmacy dataset is available from the Stanford SNAP project page at \url{https://snap.stanford.edu/decagon/}. A related public data entry for the polypharmacy side-effect network is also available at \url{https://snap.stanford.edu/biodata/datasets/10017/10017-ChChSe-Decagon.html}. The code and preprocessing scripts used to construct the experimental graphs and reproduce the reported results will be made publicly available upon acceptance.

\paragraph{Competing interests}
The authors declare that they have no competing interests.

\paragraph{Funding}
This work was supported by the Major Scientific and Technological Innovation Platform Project of Hunan Province [2024JC1003] and the Graduate Innovation Project of National University of Defense Technology [XJQY2024065].

\paragraph{Authors' contributions}
Z.Y. and Z.X. conceived the methodology, developed the software, conducted the experiments, and wrote the original draft. C.L. supervised the project, acquired funding, and reviewed \& edited the manuscript.

\paragraph{Acknowledgements}
The authors would like to express their sincere gratitude to all the referees for their careful reading and insightful suggestions.

\bibliographystyle{IEEEtran} 
\bibliography{references} 

@article{dwivedi2023benchmark,
  author  = {Vijay Prakash Dwivedi and Chaitanya K. Joshi and Anh Tuan Luu and Thomas Laurent and Yoshua Bengio and Xavier Bresson},
  title   = {Benchmarking Graph Neural Networks},
  journal = {Journal of Machine Learning Research},
  year    = {2023},
  volume  = {24},
  number  = {43},
  pages   = {1--48},
  url     = {http://jmlr.org/papers/v24/22-0567.html}
}

@inproceedings{ying2021graphormer,
  author    = {Chengxuan Ying and Tianle Cai and Shengjie Luo and Shuxin Zheng and Guolin Ke and Di He and Yanming Shen and Tie{-}Yan Liu},
  title     = {Do Transformers Really Perform Badly for Graph Representation?},
  booktitle = {Advances in Neural Information Processing Systems},
  year      = {2021},
  pages     = {28877--28888}
}

@inproceedings{rampasek2022gps,
  author    = {Ladislav Ramp{\'a}{\v{s}}ek and Michael Galkin and Vijay Prakash Dwivedi and Anh Tuan Luu and Guy Wolf and Dominique Beaini},
  title     = {Recipe for a General, Powerful, Scalable Graph Transformer},
  booktitle = {Advances in Neural Information Processing Systems},
  year      = {2022}
}

@inproceedings{li2020distance,
  author    = {Pan Li and Yanbang Wang and Hongwei Wang and Jure Leskovec},
  title     = {Distance Encoding: Design Provably More Powerful Neural Networks for Graph Representation Learning},
  booktitle = {Advances in Neural Information Processing Systems},
  year      = {2020}
}

@inproceedings{lim2023signnet,
  author    = {Derek Lim and Joshua David Robinson and Lingxiao Zhao and Tess E. Smidt and Suvrit Sra and Haggai Maron and Stefanie Jegelka},
  title     = {Sign and Basis Invariant Networks for Spectral Graph Representation Learning},
  booktitle = {International Conference on Learning Representations},
  year      = {2023}
}

@inproceedings{dwivedi2022learnable,
  author    = {Vijay Prakash Dwivedi and Anh Tuan Luu and Thomas Laurent and Yoshua Bengio and Xavier Bresson},
  title     = {Graph Neural Networks with Learnable Structural and Positional Representations},
  booktitle = {International Conference on Learning Representations},
  year      = {2022}
}

@article{gabrielsson2023rewiring,
  author  = {Rickard Br{\"u}el Gabrielsson and Mikhail Yurochkin and Justin Solomon},
  title   = {Rewiring with Positional Encodings for Graph Neural Networks},
  journal = {Transactions on Machine Learning Research},
  year    = {2023}
}

@inproceedings{canturk2024gpse,
  author    = {Semih Cant{\"u}rk and Renming Liu and Olivier Lapointe-Gagn{\'e} and Vincent L{\'e}tourneau and Guy Wolf and Dominique Beaini and Ladislav Ramp{\'a}{\v{s}}ek},
  title     = {Graph Positional and Structural Encoder},
  booktitle = {Proceedings of the 41st International Conference on Machine Learning},
  series    = {Proceedings of Machine Learning Research},
  volume    = {235},
  pages     = {5533--5566},
  year      = {2024},
  publisher = {PMLR},
  url       = {https://proceedings.mlr.press/v235/canturk24a.html}
}

@inproceedings{huang2024stability,
  author    = {Yinan Huang and William Lu and Joshua Robinson and Yu Yang and Muhan Zhang and Stefanie Jegelka and Pan Li},
  title     = {On the Stability of Expressive Positional Encodings for Graphs},
  booktitle = {International Conference on Learning Representations},
  year      = {2024}
}

@inproceedings{li2024graphtransformers,
  author    = {Hongkang Li and Meng Wang and Tengfei Ma and Sijia Liu and Zaixi Zhang and Pin{-}Yu Chen},
  title     = {What Improves the Generalization of Graph Transformers? {A} Theoretical Dive into the Self-attention and Positional Encoding},
  booktitle = {Proceedings of the 41st International Conference on Machine Learning},
  series    = {Proceedings of Machine Learning Research},
  volume    = {235},
  pages     = {28784--28829},
  year      = {2024},
  publisher = {PMLR},
  url       = {https://proceedings.mlr.press/v235/li24bo.html}
}

@inproceedings{thabet2024quantum,
  author    = {Slimane Thabet and Mehdi Djellabi and Igor Olegovich Sokolov and Sachin Kasture and Louis{-}Paul Henry and Lo{\"i}c Henriet},
  title     = {Quantum Positional Encodings for Graph Neural Networks},
  booktitle = {Proceedings of the 41st International Conference on Machine Learning},
  series    = {Proceedings of Machine Learning Research},
  volume    = {235},
  pages     = {47965--47996},
  year      = {2024},
  publisher = {PMLR},
  url       = {https://proceedings.mlr.press/v235/thabet24a.html}
}

@inproceedings{kanatsoulis2025efficientpe,
  author    = {Charilaos I. Kanatsoulis and Evelyn Choi and Stefanie Jegelka and Jure Leskovec and Alejandro Ribeiro},
  title     = {Learning Efficient Positional Encodings with Graph Neural Networks},
  booktitle = {International Conference on Learning Representations},
  year      = {2025},
  url       = {https://openreview.net/forum?id=AWg2tkbydO}
}

@article{groetschla2024benchmarkpe,
  author  = {Florian Gr{\"o}tschla and Jiaqing Xie and Roger Wattenhofer},
  title   = {Benchmarking Positional Encodings for GNNs and Graph Transformers},
  journal = {CoRR},
  volume  = {abs/2411.12732},
  year    = {2024},
  doi     = {10.48550/arXiv.2411.12732}
}

@inproceedings{keriven2023role,
  author    = {Nicolas Keriven and Samuel Vaiter},
  title     = {What Functions Can Graph Neural Networks Compute on Random Graphs? The Role of Positional Encoding},
  booktitle = {Advances in Neural Information Processing Systems},
  year      = {2023}
}

@article{zitnik2018decagon,
  author  = {Marinka Zitnik and Monica Agrawal and Jure Leskovec},
  title   = {Modeling Polypharmacy Side Effects with Graph Convolutional Networks},
  journal = {Bioinformatics},
  year    = {2018},
  volume  = {34},
  number  = {13},
  pages   = {i457--i466},
  doi     = {10.1093/bioinformatics/bty294}
}

@article{yuan2025gtsurvey,
  author  = {Chaohao Yuan and Kangfei Zhao and Ercan Engin Kuruoglu and Liang Wang and Tingyang Xu and Wenbing Huang and Deli Zhao and Hong Cheng and Yu Rong},
  title   = {A Survey of Graph Transformers: Architectures, Theories and Applications},
  journal = {CoRR},
  volume  = {abs/2502.16533},
  year    = {2025},
  doi     = {10.48550/arXiv.2502.16533}
}

@inproceedings{wang2022equivariant,
  author    = {Haorui Wang and Haoteng Yin and Muhan Zhang and Pan Li},
  title     = {Equivariant and Stable Positional Encoding for More Powerful Graph Neural Networks},
  booktitle = {International Conference on Learning Representations},
  year      = {2022}
}

@inproceedings{huang2025directed,
  author    = {Yinan Huang and Haoyu Peter Wang and Pan Li},
  title     = {What Are Good Positional Encodings for Directed Graphs?},
  booktitle = {International Conference on Learning Representations},
  year      = {2025},
  url       = {https://openreview.net/forum?id=s4Wm71LFK4}
}

@inproceedings{zhang2024spectral,
  author    = {Bohang Zhang and Lingxiao Zhao and Haggai Maron},
  title     = {On the Expressive Power of Spectral Invariant Graph Neural Networks},
  booktitle = {Proceedings of the 41st International Conference on Machine Learning},
  series    = {Proceedings of Machine Learning Research},
  volume    = {235},
  pages     = {60496--60526},
  year      = {2024},
  publisher = {PMLR},
  url       = {https://proceedings.mlr.press/v235/zhang24ck.html}
}

@inproceedings{black2024comparing,
  author    = {Mitchell Black and Zhengchao Wan and Gal Mishne and Amir Nayyeri and Yusu Wang},
  title     = {Comparing Graph Transformers via Positional Encodings},
  booktitle = {Proceedings of the 41st International Conference on Machine Learning},
  series    = {Proceedings of Machine Learning Research},
  volume    = {235},
  pages     = {4103--4139},
  year      = {2024},
  publisher = {PMLR},
  url       = {https://proceedings.mlr.press/v235/black24b.html}
}

@article{liu2023m2gcn,
  author  = {Qidong Liu and Enguang Yao and Chaoyue Liu and Xin Zhou and Yafei Li and Mingliang Xu},
  title   = {M2GCN: Multi-modal Graph Convolutional Network for Modeling Polypharmacy Side Effects},
  journal = {Applied Intelligence},
  year    = {2023},
  volume  = {53},
  number  = {6},
  pages   = {6814--6825}
}

@inproceedings{you2019pgnn,
  author    = {Jiaxuan You and Rex Ying and Jure Leskovec},
  title     = {Position-aware Graph Neural Networks},
  booktitle = {Proceedings of the 36th International Conference on Machine Learning},
  series    = {Proceedings of Machine Learning Research},
  volume    = {97},
  pages     = {7134--7143},
  year      = {2019},
  publisher = {PMLR},
  url       = {https://proceedings.mlr.press/v97/you19b.html}
}

@article{zhang2020graphbert,
  author  = {Jiawei Zhang and Haopeng Zhang and Congying Xia and Li Sun},
  title   = {Graph-Bert: Only Attention is Needed for Learning Graph Representations},
  journal = {arXiv preprint arXiv:2001.05140},
  year    = {2020},
  doi     = {10.48550/arXiv.2001.05140}
}

@article{mialon2021graphit,
  author  = {Gr{\'e}goire Mialon and Dexiong Chen and Margot Selosse and Julien Mairal},
  title   = {GraphiT: Encoding Graph Structure in Transformers},
  journal = {arXiv preprint arXiv:2106.05667},
  year    = {2021},
  doi     = {10.48550/arXiv.2106.05667}
}

@inproceedings{kreuzer2021rethinking,
  author    = {Devin Kreuzer and Dominique Beaini and William L. Hamilton and Vincent L{\'e}tourneau and Prudencio Tossou},
  title     = {Rethinking Graph Transformers with Spectral Attention},
  booktitle = {Advances in Neural Information Processing Systems},
  volume    = {34},
  pages     = {21618--21629},
  year      = {2021},
  url       = {https://proceedings.neurips.cc/paper_files/paper/2021/hash/b4fd1d2cb085390fbbadae65e07876a7-Abstract.html}
}

@article{park2022grpe,
  author  = {Wonpyo Park and Woonggi Chang and Donggeon Lee and Juntae Kim and Seung-won Hwang},
  title   = {GRPE: Relative Positional Encoding for Graph Transformer},
  journal = {arXiv preprint arXiv:2201.12787},
  year    = {2022},
  doi     = {10.48550/arXiv.2201.12787}
}

@inproceedings{ma2023grit,
  author    = {Liheng Ma and Chen Lin and Derek Lim and Adriana Romero-Soriano and Puneet K. Dokania and Mark Coates and Philip H. S. Torr and Ser-Nam Lim},
  title     = {Graph Inductive Biases in Transformers without Message Passing},
  booktitle = {Proceedings of the 40th International Conference on Machine Learning},
  series    = {Proceedings of Machine Learning Research},
  volume    = {202},
  pages     = {23321--23337},
  year      = {2023},
  publisher = {PMLR},
  url       = {https://proceedings.mlr.press/v202/ma23c.html}
}

@inproceedings{eliasof2023rfp,
  author    = {Moshe Eliasof and Fabrizio Frasca and Beatrice Bevilacqua and Eran Treister and Gal Chechik and Haggai Maron},
  title     = {Graph Positional Encoding via Random Feature Propagation},
  booktitle = {Proceedings of the 40th International Conference on Machine Learning},
  series    = {Proceedings of Machine Learning Research},
  volume    = {202},
  pages     = {9202--9223},
  year      = {2023},
  publisher = {PMLR},
  url       = {https://proceedings.mlr.press/v202/eliasof23a.html}
}

@inproceedings{he2024sheafpe,
  author    = {Yu He and Cristian Bodnar and Pietro Li{\`o}},
  title     = {Sheaf-based Positional Encodings for Graph Neural Networks},
  booktitle = {Proceedings of the 2nd NeurIPS Workshop on Symmetry and Geometry in Neural Representations},
  series    = {Proceedings of Machine Learning Research},
  volume    = {228},
  pages     = {1--18},
  year      = {2024},
  publisher = {PMLR},
  url       = {https://proceedings.mlr.press/v228/he24a.html}
}

@article{yeom2024spegt,
  author  = {Jeyoon Yeom and Taero Kim and Rakwoo Chang and Kyungwoo Song},
  title   = {Structural and Positional Ensembled Encoding for Graph Transformer},
  journal = {Pattern Recognition Letters},
  volume  = {183},
  pages   = {104--110},
  year    = {2024},
  doi     = {10.1016/j.patrec.2024.05.006}
}

@article{franks2025foundation,
  author  = {Billy Joe Franks and Moshe Eliasof and Semih Cant{\"u}rk and Guy Wolf and Carola-Bibiane Sch{\"o}nlieb and Sophie Fellenz and Marius Kloft},
  title   = {Towards Graph Foundation Models: A Study on the Generalization of Positional and Structural Encodings},
  journal = {Transactions on Machine Learning Research},
  year    = {2025},
  url     = {https://openreview.net/forum?id=mSoDRZXsqj}
}

@article{zhang2025strucgcn,
  author  = {Jie Zhang and Mingxuan Li and Yitai Xu and Hua He and Qun Li and Tao Wang},
  title   = {StrucGCN: Structural Enhanced Graph Convolutional Networks for Graph Embedding},
  journal = {Information Fusion},
  volume  = {117},
  pages   = {102893},
  year    = {2025},
  doi     = {10.1016/j.inffus.2024.102893}
}

@article{yan2026bridging,
  author  = {Zimo Yan and Zheng Xie and Runfan Duan and Chang Liu and Wumei Du},
  title   = {Bridging Distance and Spectral Positional Encodings via Anchor-Based Diffusion Geometry Approximation},
  journal = {arXiv preprint arXiv:2601.04517},
  year    = {2026},
  doi     = {10.48550/arXiv.2601.04517}
}

@article{yan2025resolving,
  author  = {Zimo Yan and Zheng Xie and Chang Liu and Yuan Wang},
  title   = {Resolving Node Identifiability in Graph Neural Processes via Laplacian Spectral Encodings},
  journal = {arXiv preprint arXiv:2511.19037},
  year    = {2025},
  doi     = {10.48550/arXiv.2511.19037}
}

\clearpage

\tableofcontents
\clearpage

\appendix
\newtheorem*{lemmaR}{Lemma}
\newtheorem*{theoremR}{Theorem}
\newtheorem*{propositionR}{Proposition}
\newtheorem*{corollaryR}{Corollary}
\newtheorem*{assumptionR}{Assumption}

\section{Proofs for Section~\ref{SE5}}
\label{app:theory-proofs}

This appendix provides the full proofs for the theoretical results in Section~\ref{SE5}. For convenience, we restate each result in the same order as in the main text before giving its proof.

\subsection{Proof of Lemma~\ref{lem:preimage-identity}}
\label{app:proof-preimage-identity}

\begin{lemmaR}[Lemma~\ref{lem:preimage-identity}]
For any fixed $(G,\mathcal A)$, let
\begin{equation}
Y \coloneqq F_{G,\mathcal A}^{(\eta)}(v_\star).
\end{equation}
Then for any measurable map
\begin{equation}
s_{G,\mathcal A}:\mathcal Y_{k,m}^{(\eta)}(G,\mathcal A)\to V,
\end{equation}
one has
\begin{equation}
\resizebox{\columnwidth}{!}{$
\begin{aligned}
\mathbb P\Big(s_{G,\mathcal A}(Y)=v_\star\ \Big|\ G,\mathcal A\Big)
&=
\mathbb E\Bigg[
\frac{
\mathbf 1\!\left\{s_{G,\mathcal A}(Y)\in \mathcal P_{G,\mathcal A}^{(\eta)}(Y)\right\}
}{
\big|\mathcal P_{G,\mathcal A}^{(\eta)}(Y)\big|
}
\ \Bigg|\ G,\mathcal A
\Bigg]\\
&\le
\mathbb E\Bigg[
\frac{1}{\big|\mathcal P_{G,\mathcal A}^{(\eta)}(v_\star)\big|}
\ \Bigg|\ G,\mathcal A
\Bigg].
\end{aligned}
$}
\end{equation}
Moreover,
\begin{equation}
\mathbb E\Bigg[
\frac{1}{\big|\mathcal P_{G,\mathcal A}^{(\eta)}(v_\star)\big|}
\ \Bigg|\ G,\mathcal A
\Bigg]
=
\frac{\big|\mathrm{Im}(F_{G,\mathcal A}^{(\eta)})\big|}{n}.
\end{equation}
The upper bound is attained by any measurable section
\begin{equation}
\tilde s_{G,\mathcal A}:\mathcal Y_{k,m}^{(\eta)}(G,\mathcal A)\to V
\end{equation}
satisfying
\begin{equation}
F_{G,\mathcal A}^{(\eta)}\big(\tilde s_{G,\mathcal A}(y)\big)=y,
\quad
\forall\, y\in \mathcal Y_{k,m}^{(\eta)}(G,\mathcal A).
\end{equation}
Consequently,
\begin{equation}
\inf_{s_{G,\mathcal A}}\ \mathrm{Err}(s_{G,\mathcal A})
=
1-\frac{\big|\mathrm{Im}(F_{G,\mathcal A}^{(\eta)})\big|}{n}.
\end{equation}
\end{lemmaR}

\begin{proof}
Let
\begin{equation}
Y\coloneqq F_{G,\mathcal A}^{(\eta)}(v_\star).
\end{equation}
Since $v_\star$ is uniformly distributed on $V$, for every
\begin{equation}
y\in \mathrm{Im}(F_{G,\mathcal A}^{(\eta)}),
\end{equation}
the conditional distribution of $v_\star$ given $Y=y$ is uniform on the preimage
\begin{equation}
\mathcal P_{G,\mathcal A}^{(\eta)}(y)
=
\{v\in V:\ F_{G,\mathcal A}^{(\eta)}(v)=y\}.
\end{equation}
Indeed, for any $v\in \mathcal P_{G,\mathcal A}^{(\eta)}(y)$,
\begin{equation}
\begin{aligned}
\mathbb P(v_\star=v\mid Y=y,G,\mathcal A)
&=
\frac{\mathbb P(v_\star=v\mid G,\mathcal A)}{\mathbb P(Y=y\mid G,\mathcal A)}\\
&=
\frac{1/n}{|\mathcal P_{G,\mathcal A}^{(\eta)}(y)|/n}
=
\frac{1}{|\mathcal P_{G,\mathcal A}^{(\eta)}(y)|}.
\end{aligned}
\end{equation}

Now fix any measurable map $s_{G,\mathcal A}$. Conditional on $Y=y$, the event
\begin{equation}
\{s_{G,\mathcal A}(Y)=v_\star\}
\end{equation}
can occur only if
\begin{equation}
s_{G,\mathcal A}(y)\in \mathcal P_{G,\mathcal A}^{(\eta)}(y).
\end{equation}
Hence
\begin{equation}
\resizebox{\columnwidth}{!}{$
\mathbb P\big(s_{G,\mathcal A}(Y)=v_\star\mid Y=y,G,\mathcal A\big)
=
\begin{cases}
\displaystyle \frac{1}{|\mathcal P_{G,\mathcal A}^{(\eta)}(y)|},
& s_{G,\mathcal A}(y)\in \mathcal P_{G,\mathcal A}^{(\eta)}(y),\\[8pt]
0,
& s_{G,\mathcal A}(y)\notin \mathcal P_{G,\mathcal A}^{(\eta)}(y).
\end{cases}
$}
\end{equation}
Equivalently,
\begin{equation}
\mathbb P\big(s_{G,\mathcal A}(Y)=v_\star\mid Y=y,G,\mathcal A\big)
=
\frac{
\mathbf 1\!\left\{s_{G,\mathcal A}(y)\in \mathcal P_{G,\mathcal A}^{(\eta)}(y)\right\}
}{
|\mathcal P_{G,\mathcal A}^{(\eta)}(y)|
}.
\end{equation}

Taking expectation with respect to $Y$ conditional on $(G,\mathcal A)$ gives
\begin{equation}
\resizebox{\columnwidth}{!}{$
\mathbb P\Big(s_{G,\mathcal A}(F_{G,\mathcal A}^{(\eta)}(v_\star))=v_\star\ \Big|\ G,\mathcal A\Big)
=
\mathbb E\Bigg[
\frac{
\mathbf 1\!\left\{s_{G,\mathcal A}(Y)\in \mathcal P_{G,\mathcal A}^{(\eta)}(Y)\right\}
}{
|\mathcal P_{G,\mathcal A}^{(\eta)}(Y)|
}
\ \Bigg|\ G,\mathcal A
\Bigg].
$}
\end{equation}
Since the indicator is at most $1$, we obtain
\begin{equation}
\mathbb P\Big(s_{G,\mathcal A}(F_{G,\mathcal A}^{(\eta)}(v_\star))=v_\star\ \Big|\ G,\mathcal A\Big)
\le
\mathbb E\Bigg[
\frac{1}{|\mathcal P_{G,\mathcal A}^{(\eta)}(Y)|}
\ \Bigg|\ G,\mathcal A
\Bigg].
\end{equation}
Because
\begin{equation}
Y=F_{G,\mathcal A}^{(\eta)}(v_\star),
\end{equation}
we also have
\begin{equation}
|\mathcal P_{G,\mathcal A}^{(\eta)}(Y)|
=
|\mathcal P_{G,\mathcal A}^{(\eta)}(v_\star)|.
\end{equation}
Thus
\begin{equation}
\mathbb P\Big(s_{G,\mathcal A}(F_{G,\mathcal A}^{(\eta)}(v_\star))=v_\star\ \Big|\ G,\mathcal A\Big)
\le
\mathbb E\Bigg[
\frac{1}{|\mathcal P_{G,\mathcal A}^{(\eta)}(v_\star)|}
\ \Bigg|\ G,\mathcal A
\Bigg].
\end{equation}

Next, let
\begin{equation}
\mathcal Y\coloneqq \mathrm{Im}(F_{G,\mathcal A}^{(\eta)}),
\quad
N_y\coloneqq |\mathcal P_{G,\mathcal A}^{(\eta)}(y)|,
\quad y\in\mathcal Y.
\end{equation}
Because $v_\star$ is uniform on $V$, we have
\begin{equation}
\mathbb P(Y=y\mid G,\mathcal A)=\frac{N_y}{n}.
\end{equation}
Hence
\begin{equation}
\begin{aligned}
\mathbb E\Bigg[
\frac{1}{|\mathcal P_{G,\mathcal A}^{(\eta)}(v_\star)|}
\ \Bigg|\ G,\mathcal A
\Bigg]
&=
\sum_{y\in\mathcal Y}
\mathbb P(Y=y\mid G,\mathcal A)\cdot \frac{1}{N_y}\\
&=
\sum_{y\in\mathcal Y}
\frac{N_y}{n}\cdot \frac{1}{N_y}\\
&=
\sum_{y\in\mathcal Y}\frac{1}{n}
=
\frac{|\mathcal Y|}{n}
=
\frac{|\mathrm{Im}(F_{G,\mathcal A}^{(\eta)})|}{n}.
\end{aligned}
\end{equation}

Now let $\tilde s_{G,\mathcal A}$ be any measurable section, namely
\begin{equation}
F_{G,\mathcal A}^{(\eta)}\big(\tilde s_{G,\mathcal A}(y)\big)=y,
\quad
\forall\, y\in\mathcal Y.
\end{equation}
Then
\begin{equation}
\tilde s_{G,\mathcal A}(y)\in \mathcal P_{G,\mathcal A}^{(\eta)}(y),
\quad
\forall\, y\in\mathcal Y,
\end{equation}
so the indicator in the previous display is identically equal to $1$. Therefore
\begin{equation}
\resizebox{\columnwidth}{!}{$
\begin{aligned}
\mathbb P\Big(\tilde s_{G,\mathcal A}(F_{G,\mathcal A}^{(\eta)}(v_\star))=v_\star\ \Big|\ G,\mathcal A\Big)
&=
\mathbb E\Bigg[
\frac{1}{|\mathcal P_{G,\mathcal A}^{(\eta)}(v_\star)|}
\ \Bigg|\ G,\mathcal A
\Bigg]\\
&=
\frac{|\mathrm{Im}(F_{G,\mathcal A}^{(\eta)})|}{n}.
\end{aligned}
$}
\end{equation}
Thus the maximal conditional success probability equals
\begin{equation}
\frac{|\mathrm{Im}(F_{G,\mathcal A}^{(\eta)})|}{n},
\end{equation}
and hence
\begin{equation}
\inf_{s_{G,\mathcal A}}\ \mathrm{Err}(s_{G,\mathcal A})
=
1-\frac{|\mathrm{Im}(F_{G,\mathcal A}^{(\eta)})|}{n}.
\end{equation}
This completes the proof.
\end{proof}

\subsection{Proof of Lemma~\ref{lem:profile-count}}
\label{app:proof-profile-count}

\begin{lemmaR}[Lemma~\ref{lem:profile-count}]
Assume $\mathrm{diam}(G)\le C_{\mathrm{diam}}\log n$. Then
\begin{equation}
\big|\{\mathbf d_{\mathcal A}(v):v\in V\}\big|\le (C\log n)^k,
\end{equation}
and in particular, the number of nonempty buckets $B_{\mathbf t}(G,\mathcal A)$ is at most $(C\log n)^k$.
\end{lemmaR}

\begin{proof}
Assume
\begin{equation}
\mathrm{diam}(G)\le C_{\mathrm{diam}}\log n.
\end{equation}
Fix any vertex $v\in V$. For each anchor $a_i\in\mathcal A$, the shortest-path distance satisfies
\begin{equation}
0\le \mathrm{SPD}(v,a_i)\le \mathrm{diam}(G)\le C_{\mathrm{diam}}\log n.
\end{equation}
Therefore each coordinate of the anchor-distance vector
\begin{equation}
\mathbf d_{\mathcal A}(v)=\big(\mathrm{SPD}(v,a_1),\dots,\mathrm{SPD}(v,a_k)\big)
\end{equation}
takes values in the finite set
\begin{equation}
\{0,1,\dots,\lfloor C_{\mathrm{diam}}\log n\rfloor\}.
\end{equation}
Hence the total number of possible $k$-dimensional distance vectors is at most
\begin{equation}
\big(\lfloor C_{\mathrm{diam}}\log n\rfloor+1\big)^k.
\end{equation}
Absorbing the additive constant into $C$, we obtain
\begin{equation}
\big|\{\mathbf d_{\mathcal A}(v):v\in V\}\big|
\le
(C\log n)^k.
\end{equation}
Since each nonempty bucket $B_{\mathbf t}(G,\mathcal A)$ corresponds to one attained vector $\mathbf t$, the number of nonempty buckets is bounded by the same quantity. This completes the proof.
\end{proof}

\subsection{Proof of Lemma~\ref{lem:image-control}}
\label{app:lem:image-control}
\begin{lemmaR}[Lemma~\ref{lem:image-control}]
Assume $\mathrm{diam}(G)\le C_{\mathrm{diam}}\log n$ and that Assumption~\ref{ass:entropy} holds on $\mathcal E_{\mathrm{ent}}$. Then on the event $\{\mathrm{diam}(G)\le C_{\mathrm{diam}}\log n\}\cap \mathcal E_{\mathrm{ent}}$,
\begin{equation}
\begin{aligned}
\big|\mathrm{Im}(F_{G,\mathcal A}^{(\eta)})\big|
&\le
(C\log n)^k\cdot \big|Q_\eta(\mathcal S_m(V))\big| \\
&\le
(C\log n)^k\cdot \Big(\frac{C_{\mathrm{ent}}}{\eta}\Big)^{c_{\mathrm{ent}} m},
\end{aligned}
\end{equation}
where $C=C(r)>0$ depends only on $C_{\mathrm{diam}}$.
\end{lemmaR}
The proof combines the distance-bucket count with the quantized spectral-code bound.

\begin{proof}
Fix $(G,\mathcal A)$ on the event
\begin{equation}
\{\mathrm{diam}(G)\le C_{\mathrm{diam}}\log n\}\cap \mathcal E_{\mathrm{ent}}.
\end{equation}
By Lemma~\ref{lem:profile-count}, the number of attained distance vectors is at most $(C\log n)^k$.

Fix one such distance vector $\mathbf t$ and consider the corresponding bucket
\begin{equation}
B_{\mathbf t}(G,\mathcal A)=\{v\in V:\ \mathbf d_{\mathcal A}(v)=\mathbf t\}.
\end{equation}
Inside this bucket, the only part of the observation map that can vary is the quantized spectral term
\begin{equation}
Q_\eta(\mathcal S_m(v)).
\end{equation}
Therefore the number of distinct observation values attained inside $B_{\mathbf t}(G,\mathcal A)$ is at most
\begin{equation}
\big|Q_\eta(\mathcal S_m(V))\big|.
\end{equation}
Summing over all attained distance buckets, we obtain
\begin{equation}
\big|\mathrm{Im}(F_{G,\mathcal A}^{(\eta)})\big|
\le
(C\log n)^k\cdot \big|Q_\eta(\mathcal S_m(V))\big|.
\end{equation}
Finally, Assumption~\ref{ass:entropy} gives
\begin{equation}
\big|Q_\eta(\mathcal S_m(V))\big|
\le
\Big(\frac{C_{\mathrm{ent}}}{\eta}\Big)^{c_{\mathrm{ent}} m},
\end{equation}
hence
\begin{equation}
\big|\mathrm{Im}(F_{G,\mathcal A}^{(\eta)})\big|
\le
(C\log n)^k\cdot \Big(\frac{C_{\mathrm{ent}}}{\eta}\Big)^{c_{\mathrm{ent}} m}.
\end{equation}
This completes the proof.
\end{proof}

\subsection{Proof of Proposition~\ref{prop:entropy-energy}}
\label{app:proof-entropy-energy}

\begin{propositionR}[Proposition~\ref{prop:entropy-energy}]
For the embedding $\mathcal S_m(v)=(\Phi_m(v,j)^2)_{j\in[m]}$ and $\eta\in(0,1)$, one has
\begin{equation}
\big|Q_\eta(\mathcal S_m(V))\big|
\le
\big|Q_\eta([0,1]^m)\big|
\le
\Big(\frac{2}{\eta}\Big)^m.
\end{equation}
In particular, Assumption~\ref{ass:entropy} holds with $C_{\mathrm{ent}}=2$, $c_{\mathrm{ent}}=1$, and $\mathcal E_{\mathrm{ent}}$ equal to the whole probability space.
\end{propositionR}

\begin{proof}
For the energy embedding,
\begin{equation}
\mathcal S_m(v)=\big(\Phi_m(v,1)^2,\dots,\Phi_m(v,m)^2\big)\in\mathbb R^m.
\end{equation}
Since $\Phi_m$ is an orthonormal eigenbasis matrix, each entry satisfies
\begin{equation}
|\Phi_m(v,j)|\le 1,
\end{equation}
and therefore each squared coordinate lies in $[0,1]$. Hence
\begin{equation}
\mathcal S_m(V)\subset [0,1]^m.
\end{equation}
Thus
\begin{equation}
\big|Q_\eta(\mathcal S_m(V))\big|
\le
\big|Q_\eta([0,1]^m)\big|.
\end{equation}

Under coordinatewise quantization with bin width $\eta$, each coordinate interval $[0,1]$ intersects at most
\begin{equation}
\left\lfloor \frac{1}{\eta}\right\rfloor+1
\end{equation}
quantization bins. Since $\eta\in(0,1)$, we have
\begin{equation}
\left\lfloor \frac{1}{\eta}\right\rfloor+1 \le \frac{2}{\eta}.
\end{equation}
Therefore
\begin{equation}
\big|Q_\eta([0,1]^m)\big|
\le
\left(\left\lfloor \frac{1}{\eta}\right\rfloor+1\right)^m
\le
\Big(\frac{2}{\eta}\Big)^m.
\end{equation}
This proves the claim.
\end{proof}

\subsection{Proof of Theorem~\ref{thm:threshold}}
\label{app:proof-threshold}

\begin{theoremR}[Theorem~\ref{thm:threshold}]
Fix $r\ge 3$ and let $G\sim\mathcal G_{n,r}$. Let $v_\star\sim\mathrm{Unif}(V)$ and let $\mathcal A=\{a_1,\dots,a_k\}$ be sampled uniformly without replacement from $V$, independently of $v_\star$ and $G$. Fix $m\ge 1$ and $\eta\in(0,1)$. Assume that the diameter bound $\mathrm{diam}(G)\le C_{\mathrm{diam}}\log n$ holds with probability $1-o(1)$ and that Assumption~\ref{ass:entropy} holds for the spectral embedding $\mathcal S_m$ at resolution $\eta$.
If there exists $\varepsilon_0\in(0,1)$ such that
\begin{equation}
k\log\log n + c_{\mathrm{ent}} m \log\!\Big(\frac{C_{\mathrm{ent}}}{\eta}\Big)\le (1-\varepsilon_0)\log n,
\end{equation}
then there exists $\varepsilon_1\in(0,\varepsilon_0)$ such that with probability $1-o(1)$ over $(G,\mathcal A)$,
\begin{equation}
\inf_{s_{G,\mathcal A}}\ \mathbb P\Big(s_{G,\mathcal A}(F_{G,\mathcal A}^{(\eta)}(v_\star))\neq v_\star\ \Big|\ G,\mathcal A\Big)
\ge
1-n^{-\varepsilon_1}.
\end{equation}
In particular, the left-hand side is bounded below by any fixed constant $\delta\in(0,1)$ for all sufficiently large $n$.
\end{theoremR}

\begin{proof}
Fix $\varepsilon_0$ as in the statement, and consider the event
\begin{equation}
\mathcal E
\coloneqq
\{\mathrm{diam}(G)\le C_{\mathrm{diam}}\log n\}\cap \mathcal E_{\mathrm{ent}}.
\end{equation}
By assumption,
\begin{equation}
\mathbb P(\mathcal E)=1-o(1).
\end{equation}
On the event $\mathcal E$, Lemma~\ref{lem:image-control} gives
\begin{equation}
\big|\mathrm{Im}(F_{G,\mathcal A}^{(\eta)})\big|
\le
(C\log n)^k\cdot \Big(\frac{C_{\mathrm{ent}}}{\eta}\Big)^{c_{\mathrm{ent}} m}.
\end{equation}
Substituting this bound into Lemma~\ref{lem:preimage-identity}, we get
\begin{equation}
\begin{aligned}
\inf_{s_{G,\mathcal A}}\ \mathrm{Err}(s_{G,\mathcal A})
&=
1-\frac{\big|\mathrm{Im}(F_{G,\mathcal A}^{(\eta)})\big|}{n}\\
&\ge
1-
\frac{(C\log n)^k\cdot \big(\frac{C_{\mathrm{ent}}}{\eta}\big)^{c_{\mathrm{ent}} m}}{n}.
\end{aligned}
\end{equation}

We now estimate the numerator. Taking logarithms,
\begin{equation}
\resizebox{\columnwidth}{!}{$
\log\left((C\log n)^k\cdot \Big(\frac{C_{\mathrm{ent}}}{\eta}\Big)^{c_{\mathrm{ent}} m}\right)
=
k\log(C\log n)+c_{\mathrm{ent}}m\log\Big(\frac{C_{\mathrm{ent}}}{\eta}\Big).
$}
\end{equation}
Since
\begin{equation}
k\log\log n \le (1-\varepsilon_0)\log n,
\end{equation}
we have
\begin{equation}
k \le \frac{(1-\varepsilon_0)\log n}{\log\log n}.
\end{equation}
Therefore
\begin{equation}
k\log C \le \frac{(1-\varepsilon_0)\log C}{\log\log n}\log n = o(\log n).
\end{equation}
Hence there exists $\varepsilon_1\in(0,\varepsilon_0)$ such that for all sufficiently large $n$,
\begin{equation}
k\log(C\log n)+c_{\mathrm{ent}}m\log\Big(\frac{C_{\mathrm{ent}}}{\eta}\Big)
\le (1-\varepsilon_1)\log n.
\end{equation}
 Thus the hypothesis implies that
\begin{equation}
(C\log n)^k\cdot \Big(\frac{C_{\mathrm{ent}}}{\eta}\Big)^{c_{\mathrm{ent}} m}
\le
n^{1-\varepsilon_1}
\end{equation}
for all sufficiently large $n$. Therefore
\begin{equation}
\inf_{s_{G,\mathcal A}}\ \mathrm{Err}(s_{G,\mathcal A})
\ge
1-\frac{n^{1-\varepsilon}}{n}
=
1-n^{-\varepsilon_1}.
\end{equation}
By definition of $\mathrm{Err}(s_{G,\mathcal A})$, this is exactly
\begin{equation}
\inf_{s_{G,\mathcal A}}\ \mathbb P\Big(s_{G,\mathcal A}(F_{G,\mathcal A}^{(\eta)}(v_\star))\neq v_\star\ \Big|\ G,\mathcal A\Big)
\ge
1-n^{-\varepsilon_1}.
\end{equation}
Since $\mathbb P(\mathcal E)=1-o(1)$, the bound holds with probability $1-o(1)$ over $(G,\mathcal A)$. This completes the proof.
\end{proof}

\subsection{Proof of Proposition~\ref{prop:avg-collision-refinement}}
\label{app:proof-avg-collision-refinement}

\begin{propositionR}[Proposition~\ref{prop:avg-collision-refinement} (restated)]
Assume $\mathrm{diam}(G)\le C_{\mathrm{diam}}\log n$. Suppose that for every attained distance vector $\mathbf t$, one of the following holds:
\begin{itemize}
    \item $|B_{\mathbf t}(G,\mathcal A)|=1$; or
    \item $|B_{\mathbf t}(G,\mathcal A)|\ge 2$, and
    \begin{equation}
    \Coll_{\eta,m}^{Q}(B_{\mathbf t}(G,\mathcal A))
    \ge
    c_0\,\eta^{c_1 m},
    \end{equation}
    together with
    \begin{equation}
    \Bal_{\eta,m}(B_{\mathbf t}(G,\mathcal A))\le \beta,
    \end{equation}
\end{itemize}
for some constants $c_0,c_1>0$ and $\beta\ge 1$. Then
\begin{equation}
\big|\mathrm{Im}(F_{G,\mathcal A}^{(\eta)})\big|
\le
D(G,\mathcal A)\cdot \Big(1+\frac{\beta}{c_0\,\eta^{c_1 m}}\Big).
\end{equation}
In particular,
\begin{equation}
\big|\mathrm{Im}(F_{G,\mathcal A}^{(\eta)})\big|
\le
(C\log n)^k\cdot \Big(1+\frac{\beta}{c_0\,\eta^{c_1 m}}\Big).
\end{equation}
\end{propositionR}

\begin{proof}
Fix $(G,\mathcal A)$ under the assumptions of the proposition. Let
\begin{equation}
\mathcal T(G,\mathcal A)\coloneqq \{\mathbf d_{\mathcal A}(v):v\in V\}
\end{equation}
denote the set of attained distance vectors, so that
\begin{equation}
D(G,\mathcal A)=|\mathcal T(G,\mathcal A)|.
\end{equation}

For each $\mathbf t\in \mathcal T(G,\mathcal A)$, write
\begin{equation}
B_{\mathbf t}\coloneqq B_{\mathbf t}(G,\mathcal A)
=
\{v\in V:\ \mathbf d_{\mathcal A}(v)=\mathbf t\},
\end{equation}
and let
\begin{equation}
b_{\mathbf t}\coloneqq |B_{\mathbf t}|.
\end{equation}
Also let
\begin{equation}
Q_\eta(\mathcal S_m(B_{\mathbf t}))
=
\{Q_\eta(\mathcal S_m(v)):\ v\in B_{\mathbf t}\}
\end{equation}
be the set of quantized spectral codes attained inside the bucket $B_{\mathbf t}$, and define
\begin{equation}
M_{\mathbf t}\coloneqq M_{\eta,m}(B_{\mathbf t})
=
\big|Q_\eta(\mathcal S_m(B_{\mathbf t}))\big|.
\end{equation}

For each
\begin{equation}
z\in Q_\eta(\mathcal S_m(B_{\mathbf t})),
\end{equation}
define the occupancy
\begin{equation}
N_{\mathbf t}(z)\coloneqq N_{\eta,m}(B_{\mathbf t};z)
=
\big|\{v\in B_{\mathbf t}:\ Q_\eta(\mathcal S_m(v))=z\}\big|.
\end{equation}
Then by construction,
\begin{equation}
\sum_{z\in Q_\eta(\mathcal S_m(B_{\mathbf t}))} N_{\mathbf t}(z)=b_{\mathbf t}.
\end{equation}

We now bound $M_{\mathbf t}$ for each bucket separately.

\medskip
\noindent
\textbf{Case 1:} $b_{\mathbf t}=1$.

In this case the bucket contains exactly one vertex, so it contributes exactly one observation value. More precisely,
\begin{equation}
M_{\mathbf t}=1.
\end{equation}

\medskip
\noindent
\textbf{Case 2:} $b_{\mathbf t}\ge 2$.

In this case, the assumptions of the proposition give
\begin{equation}
\Coll_{\eta,m}^{Q}(B_{\mathbf t})
\ge
c_0\,\eta^{c_1 m}
\end{equation}
and
\begin{equation}
\Bal_{\eta,m}(B_{\mathbf t})\le \beta.
\end{equation}

We first rewrite the balance condition. By definition,
\begin{equation}
\Bal_{\eta,m}(B_{\mathbf t})
=
\frac{M_{\mathbf t}}{b_{\mathbf t}}
\cdot
\max_{z\in Q_\eta(\mathcal S_m(B_{\mathbf t}))} N_{\mathbf t}(z).
\end{equation}
Hence
\begin{equation}
\frac{M_{\mathbf t}}{b_{\mathbf t}}
\cdot
\max_{z} N_{\mathbf t}(z)
\le
\beta,
\end{equation}
which implies
\begin{equation}
\max_{z} N_{\mathbf t}(z)
\le
\beta\,\frac{b_{\mathbf t}}{M_{\mathbf t}}.
\end{equation}

Next, we bound the quadratic occupancy sum:
\begin{equation}
\sum_z N_{\mathbf t}(z)^2
\le
\Big(\max_z N_{\mathbf t}(z)\Big)\sum_z N_{\mathbf t}(z).
\end{equation}
Using the previous inequality and the identity $\sum_z N_{\mathbf t}(z)=b_{\mathbf t}$, we obtain
\begin{equation}
\sum_z N_{\mathbf t}(z)^2
\le
\beta\,\frac{b_{\mathbf t}}{M_{\mathbf t}}\cdot b_{\mathbf t}
=
\beta\,\frac{b_{\mathbf t}^2}{M_{\mathbf t}}.
\end{equation}

On the other hand,
\begin{equation}
\begin{aligned}
\sum_z N_{\mathbf t}(z)\big(N_{\mathbf t}(z)-1\big)
&=
\sum_z N_{\mathbf t}(z)^2-\sum_z N_{\mathbf t}(z)\\
&=
\sum_z N_{\mathbf t}(z)^2-b_{\mathbf t}.
\end{aligned}
\end{equation}
Therefore
\begin{equation}
\sum_z N_{\mathbf t}(z)\big(N_{\mathbf t}(z)-1\big)
\le
\beta\,\frac{b_{\mathbf t}^2}{M_{\mathbf t}}-b_{\mathbf t}.
\end{equation}

Now we use the collision lower bound. By definition of $\Coll_{\eta,m}^{Q}(B_{\mathbf t})$,
\begin{equation}
\resizebox{\columnwidth}{!}{$
\Coll_{\eta,m}^{Q}(B_{\mathbf t})
=
\frac{1}{b_{\mathbf t}(b_{\mathbf t}-1)}
\sum_{\substack{u,v\in B_{\mathbf t}\\u\neq v}}
\mathbf 1\Big\{Q_\eta(\mathcal S_m(u))=Q_\eta(\mathcal S_m(v))\Big\}.
$}
\end{equation}
Grouping the ordered pairs $(u,v)$ according to their common quantized code yields
\begin{equation}
\resizebox{\columnwidth}{!}{$
\sum_{\substack{u,v\in B_{\mathbf t}\\u\neq v}}
\mathbf 1\Big\{Q_\eta(\mathcal S_m(u))=Q_\eta(\mathcal S_m(v))\Big\}
=
\sum_z N_{\mathbf t}(z)\big(N_{\mathbf t}(z)-1\big).
$}
\end{equation}
Hence
\begin{equation}
\sum_z N_{\mathbf t}(z)\big(N_{\mathbf t}(z)-1\big)
=
b_{\mathbf t}(b_{\mathbf t}-1)\Coll_{\eta,m}^{Q}(B_{\mathbf t}).
\end{equation}
Using the lower bound on $\Coll_{\eta,m}^{Q}(B_{\mathbf t})$, we obtain
\begin{equation}
\sum_z N_{\mathbf t}(z)\big(N_{\mathbf t}(z)-1\big)
\ge
b_{\mathbf t}(b_{\mathbf t}-1)c_0\,\eta^{c_1 m}.
\end{equation}

Combining the upper and lower bounds on the same quantity gives
\begin{equation}
b_{\mathbf t}(b_{\mathbf t}-1)c_0\,\eta^{c_1 m}
\le
\beta\,\frac{b_{\mathbf t}^2}{M_{\mathbf t}}-b_{\mathbf t}.
\end{equation}
Add $b_{\mathbf t}$ to both sides:
\begin{equation}
b_{\mathbf t}+b_{\mathbf t}(b_{\mathbf t}-1)c_0\,\eta^{c_1 m}
\le
\beta\,\frac{b_{\mathbf t}^2}{M_{\mathbf t}}.
\end{equation}
Factor out $b_{\mathbf t}$ on the left-hand side:
\begin{equation}
b_{\mathbf t}\Big(1+(b_{\mathbf t}-1)c_0\,\eta^{c_1 m}\Big)
\le
\beta\,\frac{b_{\mathbf t}^2}{M_{\mathbf t}}.
\end{equation}
Since $b_{\mathbf t}\ge 1$, we can divide both sides by $b_{\mathbf t}$ and obtain
\begin{equation}
1+(b_{\mathbf t}-1)c_0\,\eta^{c_1 m}
\le
\beta\,\frac{b_{\mathbf t}}{M_{\mathbf t}}.
\end{equation}
Equivalently,
\begin{equation}
M_{\mathbf t}
\le
\beta\,
\frac{b_{\mathbf t}}{1+(b_{\mathbf t}-1)c_0\,\eta^{c_1 m}}.
\end{equation}

We now simplify this bound. Since $\Coll_{\eta,m}^{Q}(B_{\mathbf t})\le 1$ by definition, and since
\begin{equation}
\Coll_{\eta,m}^{Q}(B_{\mathbf t})\ge c_0\,\eta^{c_1 m},
\end{equation}
it follows that
\begin{equation}
c_0\,\eta^{c_1 m}\le 1.
\end{equation}
Therefore
\begin{equation}
1-c_0\,\eta^{c_1 m}\ge 0.
\end{equation}
Using this, we have
\begin{equation}
1+(b_{\mathbf t}-1)c_0\,\eta^{c_1 m}
=
b_{\mathbf t}c_0\,\eta^{c_1 m}+\big(1-c_0\,\eta^{c_1 m}\big)
\ge
b_{\mathbf t}c_0\,\eta^{c_1 m}.
\end{equation}
Substituting this lower bound into the denominator yields
\begin{equation}
M_{\mathbf t}
\le
\beta\,
\frac{b_{\mathbf t}}{b_{\mathbf t}c_0\,\eta^{c_1 m}}
=
\frac{\beta}{c_0\,\eta^{c_1 m}}.
\end{equation}

Thus, for every attained distance vector $\mathbf t$, we have shown that
\begin{equation}
M_{\mathbf t}\le 1
\quad\text{if } b_{\mathbf t}=1,
\end{equation}
and
\begin{equation}
M_{\mathbf t}\le \frac{\beta}{c_0\,\eta^{c_1 m}}
\quad\text{if } b_{\mathbf t}\ge 2.
\end{equation}

Now sum over all attained distance vectors. Let
\begin{equation}
\mathcal T_1\coloneqq \{\mathbf t\in\mathcal T(G,\mathcal A):\ b_{\mathbf t}=1\}
\end{equation}
and
\begin{equation}
\mathcal T_{\ge 2}\coloneqq \{\mathbf t\in\mathcal T(G,\mathcal A):\ b_{\mathbf t}\ge 2\}.
\end{equation}
Then
\begin{equation}
\mathcal T(G,\mathcal A)=\mathcal T_1\cup \mathcal T_{\ge 2},
\quad
\mathcal T_1\cap \mathcal T_{\ge 2}=\varnothing,
\end{equation}
and
\begin{equation}
|\mathcal T_1|+|\mathcal T_{\ge 2}|=D(G,\mathcal A).
\end{equation}

The total image size is the sum of the numbers of distinct quantized spectral codes over all attained distance buckets:
\begin{equation}
\big|\mathrm{Im}(F_{G,\mathcal A}^{(\eta)})\big|
=
\sum_{\mathbf t\in\mathcal T(G,\mathcal A)} M_{\mathbf t}.
\end{equation}
Hence
\begin{equation}
\begin{aligned}
\big|\mathrm{Im}(F_{G,\mathcal A}^{(\eta)})\big|
&=
\sum_{\mathbf t\in\mathcal T_1} M_{\mathbf t}
+
\sum_{\mathbf t\in\mathcal T_{\ge 2}} M_{\mathbf t} \\
&\le
\sum_{\mathbf t\in\mathcal T_1} 1
+
\sum_{\mathbf t\in\mathcal T_{\ge 2}}
\frac{\beta}{c_0\,\eta^{c_1 m}} \\
&=
|\mathcal T_1|
+
|\mathcal T_{\ge 2}|
\frac{\beta}{c_0\,\eta^{c_1 m}} \\
&\le
D(G,\mathcal A)\cdot \Big(1+\frac{\beta}{c_0\,\eta^{c_1 m}}\Big).
\end{aligned}
\end{equation}
This proves the first bound.

Finally, Lemma~\ref{lem:profile-count} implies
\begin{equation}
D(G,\mathcal A)\le (C\log n)^k,
\end{equation}
and therefore
\begin{equation}
\big|\mathrm{Im}(F_{G,\mathcal A}^{(\eta)})\big|
\le
(C\log n)^k\cdot \Big(1+\frac{\beta}{c_0\,\eta^{c_1 m}}\Big).
\end{equation}
This completes the proof.
\end{proof}

\subsection{Proof of Corollary~\ref{cor:image-control-refined}}
\label{app:proof-image-control-refined}

\begin{corollaryR}[Corollary~\ref{cor:image-control-refined} (restated)]
Assume $\mathrm{diam}(G)\le C_{\mathrm{diam}}\log n$ with probability $1-o(1)$, and suppose there exists an event $\mathcal E_{\mathrm{ref}}$ with $\mathbb P(\mathcal E_{\mathrm{ref}})=1-o(1)$ such that on $\mathcal E_{\mathrm{ref}}$, for every attained distance vector $\mathbf t$, one of the following holds:
\begin{itemize}
    \item $|B_{\mathbf t}(G,\mathcal A)|=1$; or
    \item $|B_{\mathbf t}(G,\mathcal A)|\ge 2$, and
    \begin{equation}
    \Coll_{\eta,m}^{Q}(B_{\mathbf t}(G,\mathcal A))
    \ge
    c_0\,\eta^{c_1 m},
    \end{equation}
    together with
    \begin{equation}
    \Bal_{\eta,m}(B_{\mathbf t}(G,\mathcal A))\le \beta,
    \end{equation}
\end{itemize}
for some constants $c_0,c_1>0$ and $\beta\ge 1$. Then with probability $1-o(1)$ over $(G,\mathcal A)$,
\begin{equation}
\big|\mathrm{Im}(F_{G,\mathcal A}^{(\eta)})\big|
\le
(C\log n)^k\cdot \Big(1+\frac{\beta}{c_0\,\eta^{c_1 m}}\Big).
\end{equation}
In particular, since $\eta\in(0,1)$,
\begin{equation}
\big|\mathrm{Im}(F_{G,\mathcal A}^{(\eta)})\big|
\le
(C\log n)^k\cdot \Big(1+\frac{\beta}{c_0}\Big)\eta^{-c_1 m}.
\end{equation}
\end{corollaryR}

\begin{proof}
Consider the event
\begin{equation}
\mathcal E
\coloneqq
\{\mathrm{diam}(G)\le C_{\mathrm{diam}}\log n\}\cap \mathcal E_{\mathrm{ref}}.
\end{equation}
By assumption,
\begin{equation}
\mathbb P(\mathcal E)=1-o(1).
\end{equation}

Fix $(G,\mathcal A)$ on the event $\mathcal E$. Then the diameter bound
\begin{equation}
\mathrm{diam}(G)\le C_{\mathrm{diam}}\log n
\end{equation}
holds, and for every attained distance vector $\mathbf t$, the bucketwise assumptions appearing in Proposition~\ref{prop:avg-collision-refinement} are satisfied. Therefore Proposition~\ref{prop:avg-collision-refinement} applies and yields
\begin{equation}
\big|\mathrm{Im}(F_{G,\mathcal A}^{(\eta)})\big|
\le
D(G,\mathcal A)\cdot \Big(1+\frac{\beta}{c_0\,\eta^{c_1 m}}\Big).
\end{equation}
Now apply Lemma~\ref{lem:profile-count}, which gives
\begin{equation}
D(G,\mathcal A)\le (C\log n)^k.
\end{equation}
Hence
\begin{equation}
\big|\mathrm{Im}(F_{G,\mathcal A}^{(\eta)})\big|
\le
(C\log n)^k\cdot \Big(1+\frac{\beta}{c_0\,\eta^{c_1 m}}\Big).
\end{equation}
This proves the first conclusion.

For the second conclusion, note that $\eta\in(0,1)$ implies
\begin{equation}
\eta^{-c_1 m}\ge 1.
\end{equation}
Therefore
\begin{equation}
1\le \eta^{-c_1 m}
\end{equation}
and also
\begin{equation}
\frac{\beta}{c_0\,\eta^{c_1 m}}
=
\frac{\beta}{c_0}\eta^{-c_1 m}.
\end{equation}
Adding these two bounds gives
\begin{equation}
1+\frac{\beta}{c_0\,\eta^{c_1 m}}
\le
\Big(1+\frac{\beta}{c_0}\Big)\eta^{-c_1 m}.
\end{equation}
Substituting into the previous estimate yields
\begin{equation}
\big|\mathrm{Im}(F_{G,\mathcal A}^{(\eta)})\big|
\le
(C\log n)^k\cdot \Big(1+\frac{\beta}{c_0}\Big)\eta^{-c_1 m}.
\end{equation}
Since $\mathbb P(\mathcal E)=1-o(1)$, both bounds hold with probability $1-o(1)$ over $(G,\mathcal A)$. This completes the proof.
\end{proof}

\subsection{Proof of Corollary~\ref{cor:threshold-refined}}
\label{app:proof-threshold-refined}

\begin{corollaryR}[Corollary~\ref{cor:threshold-refined} (restated)]
In the setting of Theorem~\ref{thm:threshold}, further assume that there exist constants $c_0,c_1>0$ and $\beta\ge 1$, and an event $\mathcal E_{\mathrm{ref}}$ with $\mathbb P(\mathcal E_{\mathrm{ref}})=1-o(1)$ such that on $\mathcal E_{\mathrm{ref}}$, for every attained distance vector $\mathbf t$, one of the following holds:
\begin{itemize}
    \item $|B_{\mathbf t}(G,\mathcal A)|=1$; or
    \item $|B_{\mathbf t}(G,\mathcal A)|\ge 2$, and
    \begin{equation}
    \Coll_{\eta,m}^{Q}(B_{\mathbf t}(G,\mathcal A))
    \ge
    c_0\,\eta^{c_1 m},
    \end{equation}
    together with
    \begin{equation}
    \Bal_{\eta,m}(B_{\mathbf t}(G,\mathcal A))\le \beta.
    \end{equation}
\end{itemize}
If there exists $\varepsilon_0\in(0,1)$ such that
\begin{equation}
k\log\log n + c_1 m \log(1/\eta)\le (1-\varepsilon_0)\log n,
\end{equation}
then there exists $\varepsilon_1\in(0,\varepsilon_0)$ such that with probability $1-o(1)$ over $(G,\mathcal A)$,
\begin{equation}
\inf_{s_{G,\mathcal A}}\ \mathbb P\Big(s_{G,\mathcal A}(F_{G,\mathcal A}^{(\eta)}(v_\star))\neq v_\star\ \Big|\ G,\mathcal A\Big)
\ge
1-n^{-\varepsilon_1}.
\end{equation}
\end{corollaryR}

\begin{proof}
Consider the event
\begin{equation}
\mathcal E
\coloneqq
\{\mathrm{diam}(G)\le C_{\mathrm{diam}}\log n\}\cap \mathcal E_{\mathrm{ref}}.
\end{equation}
By the assumptions of Theorem~\ref{thm:threshold} and of the present corollary,
\begin{equation}
\mathbb P(\mathcal E)=1-o(1).
\end{equation}

Fix $(G,\mathcal A)$ on the event $\mathcal E$. Then Corollary~\ref{cor:image-control-refined} applies and yields
\begin{equation}
\big|\mathrm{Im}(F_{G,\mathcal A}^{(\eta)})\big|
\le
(C\log n)^k\cdot \Big(1+\frac{\beta}{c_0}\Big)\eta^{-c_1 m}.
\end{equation}
Substituting this estimate into Lemma~\ref{lem:preimage-identity}, we obtain
\begin{equation}
\begin{aligned}
\inf_{s_{G,\mathcal A}}\ \mathrm{Err}(s_{G,\mathcal A})
&=
1-\frac{\big|\mathrm{Im}(F_{G,\mathcal A}^{(\eta)})\big|}{n}\\
&\ge
1-
\frac{(C\log n)^k\cdot \big(1+\frac{\beta}{c_0}\big)\eta^{-c_1 m}}{n}.
\end{aligned}
\end{equation}

We now estimate the numerator. Taking logarithms gives
\begin{equation}
\begin{aligned}
&\log\left((C\log n)^k\cdot \Big(1+\frac{\beta}{c_0}\Big)\eta^{-c_1 m}\right)\\
=&
k\log(C\log n)+\log\Big(1+\frac{\beta}{c_0}\Big)+c_1 m\log(1/\eta).
\end{aligned}
\end{equation}
Using
\begin{equation}
\log(C\log n)=\log\log n+\log C,
\end{equation}
we can rewrite this as
\begin{equation}
k\log\log n + c_1 m\log(1/\eta) + k\log C + \log\Big(1+\frac{\beta}{c_0}\Big).
\end{equation}

By assumption,
\begin{equation}
k\log\log n + c_1 m\log(1/\eta)\le (1-\varepsilon_0)\log n.
\end{equation}
In particular,
\begin{equation}
k\log\log n\le (1-\varepsilon_0)\log n,
\end{equation}
and therefore
\begin{equation}
k\le \frac{(1-\varepsilon_0)\log n}{\log\log n}.
\end{equation}
Multiplying by the constant $\log C$ yields
\begin{equation}
k\log C
\le
\frac{(1-\varepsilon_0)\log C}{\log\log n}\log n
=
o(\log n).
\end{equation}
Also,
\begin{equation}
\log\Big(1+\frac{\beta}{c_0}\Big)=O(1)=o(\log n).
\end{equation}
Hence the two extra terms
\begin{equation}
k\log C + \log\Big(1+\frac{\beta}{c_0}\Big)
\end{equation}
are negligible compared with $\log n$. Therefore there exists $\varepsilon_1\in(0,\varepsilon_0)$ such that for all sufficiently large $n$,
\begin{equation}
k\log(C\log n)+\log\Big(1+\frac{\beta}{c_0}\Big)+c_1 m\log(1/\eta)
\le
(1-\varepsilon_1)\log n.
\end{equation}
Equivalently,
\begin{equation}
(C\log n)^k\cdot \Big(1+\frac{\beta}{c_0}\Big)\eta^{-c_1 m}
\le
n^{1-\varepsilon_1}
\end{equation}
for all sufficiently large $n$.

Substituting this into the previous lower bound on the error gives
\begin{equation}
\inf_{s_{G,\mathcal A}}\ \mathrm{Err}(s_{G,\mathcal A})
\ge
1-\frac{n^{1-\varepsilon_1}}{n}
=
1-n^{-\varepsilon_1}.
\end{equation}
By the definition of $\mathrm{Err}(s_{G,\mathcal A})$, this is exactly
\begin{equation}
\inf_{s_{G,\mathcal A}}\ \mathbb P\Big(s_{G,\mathcal A}(F_{G,\mathcal A}^{(\eta)}(v_\star))\neq v_\star\ \Big|\ G,\mathcal A\Big)
\ge
1-n^{-\varepsilon_1}.
\end{equation}
Since $\mathbb P(\mathcal E)=1-o(1)$, the conclusion holds with probability $1-o(1)$ over $(G,\mathcal A)$. This completes the proof.
\end{proof}

\section{Detailed Experimental Protocol}
\label{app:exp-details}

\subsection{Datasets}

We evaluate the proposed encoding on both synthetic and real-world graphs.

For the synthetic study, we use random 3-regular graphs to examine the phase-transition-like behavior predicted qualitatively by the converse theory. The case $m=0$ is treated as a distance-only baseline outside the formal spectral statement.
The parameter grid is
\begin{equation*}
\begin{aligned}
n &\in \{500,1000,2000,4000\},\\
k &\in \{1,2,3,4,6,8\},\\
m &\in \{0,1,2,5\},\\
\eta &\in \{0.9,0.7,0.5,0.3,0.1\},
\end{aligned}
\end{equation*}
with 20 independent graph trials for each configuration.

For the real-world study, we consider two DDI graph constructions. 
The DrugBank graph is built from positive interaction pairs in the processed DrugBank train/test CSV files and then restricted to its largest connected component, resulting in
\begin{equation*}
\begin{aligned}
N_{\mathrm{DB}} = 1684,
\quad
E_{\mathrm{DB}} = 189774.
\end{aligned}
\end{equation*}
The Decagon-derived graph is constructed in the same way from the processed Decagon splits, resulting in
\begin{equation*}
\begin{aligned}
N_{\mathrm{DEC}} = 40,
\quad
E_{\mathrm{DEC}} = 466.
\end{aligned}
\end{equation*}
For the real-graph experiments, we evaluate
\begin{equation*}
\begin{aligned}
k &\in \{1,2,3,4,6,8\},\\
m &\in \{0,2,5,10\},\\
\eta &\in \{0.9,0.5,0.25,0.1\},
\end{aligned}
\end{equation*}
and all reported quantities are averaged over 30 independent random anchor samplings.

\subsection{Baselines}

Since our goal is to evaluate the identifiability of the encoding itself rather than a learned predictor, the baselines are defined at the feature level:
\begin{equation}
\begin{aligned}
\text{NoPE}: \quad &F(v)=0,\\
\text{Distance}: \quad &F(v)=d_A(v),\\
\text{Spectral}: \quad &F(v)=Q_{\eta}(S_m(v)),\\
\text{Full}: \quad &F(v)=\bigl(d_A(v),Q_{\eta}(S_m(v))\bigr).
\end{aligned}
\end{equation}

Under the full feature, we also compare three anchor selection strategies: random sampling without replacement, degree-based selection, and a farthest-point greedy strategy.

\subsection{Encoding Configuration}

The experiments do not involve end-to-end learning. 
Instead, the object under evaluation is the deterministic observation map
\begin{equation}\label{eq43}
\begin{aligned}
F(v)=\bigl(d_A(v),Q_{\eta}(S_m(v))\bigr).
\end{aligned}
\end{equation}

For each anchor set $A=\{a_1,\ldots,a_k\}$, the distance component is
\begin{equation}
\begin{aligned}
d_A(v)
=
\bigl(
\mathrm{SPD}(v,a_1),\ldots,\mathrm{SPD}(v,a_k)
\bigr).
\end{aligned}
\end{equation}
The spectral component is constructed from the normalized Laplacian
\begin{equation}
\begin{aligned}
\mathcal{L}
=
I-D^{-1/2}AD^{-1/2},
\end{aligned}
\end{equation}
using the first $m$ non-trivial eigenvectors:
\begin{equation}
\begin{aligned}
S_m(v)
=
\bigl(
\Phi_m(v,1)^2,\ldots,\Phi_m(v,m)^2
\bigr).
\end{aligned}
\end{equation}

In all experiments, we use relative quantization. 
Given the spectral matrix $S_m$, we define
\begin{equation}\label{eq:rel-quant-step}
\begin{aligned}
\Delta_{\eta}
=
\eta \cdot \max_{i,j}|(S_m)_{ij}|,
\end{aligned}
\end{equation}
and quantize each coordinate by
\begin{equation}\label{eq:rel-quant-rule}
\begin{aligned}
\bigl(Q_{\eta}(x)\bigr)_j
=
\Delta_{\eta}
\cdot
\mathrm{round}
\left(
\frac{x_j}{\Delta_{\eta}}
\right),
\quad j\in[m].
\end{aligned}
\end{equation}
Thus, smaller $\eta$ corresponds to finer relative quantization. 
Unless otherwise stated, all main-text experiments use the relative quantization rule in \eqref{eq:rel-quant-step}-\eqref{eq:rel-quant-rule}. Absolute quantization is discussed separately in this appendix as a calibration and robustness check.

For the synthetic phase-transition study, results are organized by
\begin{equation}
\begin{aligned}
\rho_{\mathrm{eng}}
=
\frac{
k\log\log n + m\log\!\bigl(C_{\mathrm{ent}}/\eta\bigr)
}{
\log n
},
\quad
C_{\mathrm{ent}}=2.
\end{aligned}
\end{equation}

\subsection{Evaluation Measures}

For a fixed graph and parameter configuration, we measure the image size of the observation map and define
\begin{equation}\label{eq50}
\begin{aligned}
\mathrm{succ}(F)
&=
\frac{|\mathrm{Im}(F)|}{n},
\\
\mathrm{err}(F)
&=
1-\mathrm{succ}(F).
\end{aligned}
\end{equation}
A smaller error means that more vertices receive unique structural codes.

For the synthetic experiments, we also define the empirical critical number of anchors
\begin{equation}
\begin{aligned}
k_{\mathrm{emp}}(n,m,\eta)
=
\min
\left\{
k:
\mathbb{E}[\mathrm{err}(F)]\le 0.1
\right\}.
\end{aligned}
\end{equation}

For the synthetic experiments, expectations are approximated by averaging over 20 independently sampled graph instances for each configuration.
For the real-graph experiments, the graph is fixed and repeated trials correspond to 30 independent random anchor samplings.

\subsection{Ablation Protocol}

To understand the role of each component, we conduct ablations on random $3$-regular graphs with
\begin{equation*}
\begin{aligned}
n &\in \{500,1000,2000\},\\
k &\in \{1,2,3,4,6,8\},\\
m &\in \{0,2,5,10\},\\
\eta &\in \{0.9,0.5,0.25,0.1\},
\end{aligned}
\end{equation*}
using 20 independent graph trials.

The ablation study includes three parts: comparison of the four feature variants, sensitivity analysis over $(k,m,\eta)$ under the full feature, and comparison of the three anchor selection strategies.

For details on datasets, baselines, and hyperparameter configurations, please refer to \url{https://github.com/yzz980314/Graph_Node_Identifiability_via_Observation_Maps}.

\section{Supplementary Analysis and Detailed Results}
\label{sec:supplementary_analysis}

\subsection{Complete Diagnostic Tables for the Observation Map}
\label{app:full_threshold_tables}

For each configuration \((n,m,\eta)\), we define \(k_{\mathrm{emp}}\) as the smallest anchor number such that the mean error rate is at most \(0.1\). 
We also record the corresponding theory-guided ratio \(\rho_{\mathrm{emp}}\) and the spectral code complexity \(|Q_\eta(\widetilde S_m(V))|\), where
\begin{equation*}
    \begin{aligned}
       F(v)&=\bigl(d_A(v),\,Q_\eta(\widetilde S_m(v))\bigr),\\
    \widetilde S_m(v)&=n\bigl(\phi_1(v)^2,\ldots,\phi_m(v)^2\bigr). 
    \end{aligned}
\end{equation*}

The main threshold patterns are simple:  
(i) the distance-only baseline is stable, with \(k_{\mathrm{emp}}=6\) for all tested \(n\) and \(\eta\);  
(ii) increasing \(m\) and refining \(\eta\) generally reduce \(k_{\mathrm{emp}}\);  
(iii) the reduction in \(k_{\mathrm{emp}}\) is accompanied by a sharp growth of \(|Q_\eta(\widetilde S_m(V))|\).

\begin{table}[t]
\centering
\small
\caption{Full empirical threshold table \(k_{\mathrm{emp}}\).}
\label{tab:appendix_full_kemp}
\resizebox{\linewidth}{!}{%
\begin{tabular}{ccccccc}
\toprule
\toprule
$n$ & $m$ & $\eta=0.9$ & $\eta=0.7$ & $\eta=0.5$ & $\eta=0.3$ & $\eta=0.1$ \\
\midrule
\multirow{4}{*}{500}
 & 0 & 6 & 6 & 6 & 6 & 6 \\
 & 1 & 4 & 4 & 4 & 4 & 3 \\
 & 2 & 4 & 4 & 3 & 3 & 2 \\
 & 5 & 3 & 2 & 1 & 1 & 1 \\
\specialrule{0.08em}{0.15em}{0.15em}

\multirow{4}{*}{1000}
 & 0 & 6 & 6 & 6 & 6 & 6 \\
 & 1 & 6 & 6 & 6 & 4 & 4 \\
 & 2 & 4 & 4 & 4 & 3 & 2 \\
 & 5 & 3 & 2 & 2 & 1 & 1 \\
\specialrule{0.08em}{0.15em}{0.15em}

\multirow{4}{*}{2000}
 & 0 & 6 & 6 & 6 & 6 & 6 \\
 & 1 & 6 & 6 & 6 & 6 & 4 \\
 & 2 & 6 & 4 & 4 & 4 & 3 \\
 & 5 & 3 & 3 & 2 & 1 & 1 \\
\specialrule{0.08em}{0.15em}{0.15em}

\multirow{4}{*}{4000}
 & 0 & 6 & 6 & 6 & 6 & 6 \\
 & 1 & 6 & 6 & 6 & 6 & 6 \\
 & 2 & 6 & 6 & 6 & 4 & 3 \\
 & 5 & 4 & 3 & 2 & 1 & 1 \\
\bottomrule
\bottomrule
\end{tabular}%
}
\end{table}

\begin{table}[t]
\centering
\small
\caption{Theory-guided threshold ratio \(\rho_{\mathrm{emp}}\).}
\label{tab:appendix_full_rho}
\resizebox{\linewidth}{!}{%
\begin{tabular}{ccrrrrr}
\toprule
\toprule
$n$ & $m$ & $\eta=0.9$ & $\eta=0.7$ & $\eta=0.5$ & $\eta=0.3$ & $\eta=0.1$ \\
\midrule
\multirow{4}{*}{500}
 & 0 & 1.764 & 1.764 & 1.764 & 1.764 & 1.764 \\
 & 1 & 1.304 & 1.345 & 1.399 & 1.481 & 1.364 \\
 & 2 & 1.433 & 1.514 & 1.328 & 1.492 & 1.552 \\
 & 5 & 1.524 & 1.433 & 1.409 & 1.820 & 2.704 \\
\specialrule{0.08em}{0.15em}{0.15em}

\multirow{4}{*}{1000}
 & 0 & 1.679 & 1.679 & 1.679 & 1.679 & 1.679 \\
 & 1 & 1.794 & 1.831 & 1.879 & 1.394 & 1.553 \\
 & 2 & 1.350 & 1.423 & 1.520 & 1.389 & 1.427 \\
 & 5 & 1.417 & 1.319 & 1.563 & 1.653 & 2.448 \\
\specialrule{0.08em}{0.15em}{0.15em}

\multirow{4}{*}{2000}
 & 0 & 1.601 & 1.601 & 1.601 & 1.601 & 1.601 \\
 & 1 & 1.706 & 1.739 & 1.783 & 1.851 & 1.462 \\
 & 2 & 1.811 & 1.344 & 1.432 & 1.567 & 1.589 \\
 & 5 & 1.326 & 1.491 & 1.446 & 1.515 & 2.237 \\
\specialrule{0.08em}{0.15em}{0.15em}

\multirow{4}{*}{4000}
 & 0 & 1.530 & 1.530 & 1.530 & 1.530 & 1.530 \\
 & 1 & 1.627 & 1.657 & 1.698 & 1.759 & 1.892 \\
 & 2 & 1.723 & 1.784 & 1.865 & 1.478 & 1.488 \\
 & 5 & 1.502 & 1.398 & 1.346 & 1.399 & 2.061 \\
\bottomrule
\bottomrule
\end{tabular}%
}
\end{table}

\begin{table}[t]
\centering
\small
\caption{Spectral code complexity \(|Q_\eta(\widetilde S_m(V))|\) at threshold.}
\label{tab:appendix_full_spectral}
\resizebox{\linewidth}{!}{%
\begin{tabular}{ccrrrrr}
\toprule
\toprule
$n$ & $m$ & $\eta=0.9$ & $\eta=0.7$ & $\eta=0.5$ & $\eta=0.3$ & $\eta=0.1$ \\
\midrule
\multirow{4}{*}{500}
 & 0 & 1.0 & 1.0 & 1.0 & 1.0 & 1.0 \\
 & 1 & 9.1 & 11.7 & 15.7 & 23.8 & 51.8 \\
 & 2 & 31.2 & 46.6 & 68.7 & 117.7 & 285.6 \\
 & 5 & 213.2 & 290.6 & 395.6 & 469.8 & 499.5 \\
\specialrule{0.08em}{0.15em}{0.15em}

\multirow{4}{*}{1000}
 & 0 & 1.0 & 1.0 & 1.0 & 1.0 & 1.0 \\
 & 1 & 10.5 & 14.4 & 17.5 & 27.1 & 66.0 \\
 & 2 & 46.2 & 64.5 & 97.2 & 169.3 & 444.2 \\
 & 5 & 345.2 & 491.6 & 693.9 & 900.9 & 997.8 \\
\specialrule{0.08em}{0.15em}{0.15em}

\multirow{4}{*}{2000}
 & 0 & 1.0 & 1.0 & 1.0 & 1.0 & 1.0 \\
 & 1 & 12.4 & 16.0 & 22.0 & 32.5 & 76.3 \\
 & 2 & 59.6 & 78.0 & 128.6 & 240.5 & 692.8 \\
 & 5 & 561.4 & 820.3 & 1196.5 & 1677.7 & 1990.1 \\
\specialrule{0.08em}{0.15em}{0.15em}

\multirow{4}{*}{4000}
 & 0 & 1.0 & 1.0 & 1.0 & 1.0 & 1.0 \\
 & 1 & 14.4 & 21.1 & 25.1 & 37.1 & 87.7 \\
 & 2 & 74.2 & 109.8 & 163.6 & 315.2 & 1033.6 \\
 & 5 & 898.6 & 1345.2 & 2091.1 & 3123.6 & 3962.0 \\
\bottomrule
\bottomrule
\end{tabular}%
}
\end{table}

\section{Additional ablation details}
\label{app:ablation_details}

\subsection{Quantization calibration.}
The main text uses relative quantization throughout. In this appendix, we additionally compare relative and absolute quantization as a calibration and robustness check. Relative quantization remains nearly injective throughout the tested range, with mean error at most $6.7\times 10^{-5}$ and spectral-code ratio at least $0.9933$. For absolute quantization, however, performance depends strongly on the calibration of the step size. As $\eta$ decreases from $5\times 10^{-3}$ to $2\times 10^{-4}$, the mean error drops monotonically from $0.1686$ to $10^{-4}$, while the occupied spectral-code ratio increases from $0.0593$ to $0.9473$. For $\eta \le 10^{-4}$, absolute quantization also becomes essentially injective. This confirms that the qualitative identifiability trend does not rely on the relative-quantization implementation alone, but persists under a properly calibrated absolute quantization scheme.

\subsection{Bucketwise refinement diagnostics.}
Figure~\ref{fig:app_bucketwise} reports bucketwise refinement statistics across configurations. For an intermediate regime such as $(n,k,m,\eta)=(2000,4,5,0.25)$, the error is still $0.1521$, the singleton-node fraction is only $0.3803$, and the mean within-bucket collision level remains $0.2603$. By contrast, in a near-injective regime such as $(2000,6,10,0.1)$, the error drops to $10^{-4}$, the singleton-node fraction rises to $0.9422$, and the mean within-bucket collision level becomes negligible ($0.0039$). Thus, the bucketwise statistics are most informative in the intermediate regime, while they become largely vacuous once the distance partition is already almost fully separating.

\subsection{Additional robustness detail.}
For completeness, under repeated random-anchor resampling, the mean within-graph error range is $0.0848$ for $n=1000$ and $0.1189$ for $n=2000$, with maximum observed ranges $0.1870$ and $0.2120$, respectively. These fluctuations are modest relative to the graph-to-graph variation and do not alter the overall trends.

\begin{table}[t]
\centering
\scriptsize
\setlength{\tabcolsep}{4pt}
\caption{Calibration of the absolute quantization step. Relative quantization remains nearly injective throughout the tested range; the table reports the detailed absolute-quantization trend.}
\label{tab:app_ablation_quantization}
\begin{tabular}{cccc}
\toprule
$\eta$ & Mean error & Mean preimage & Spectral-code ratio \\
\midrule
$5\times 10^{-3}$ & 0.1686   & 1.2464   & 0.0593 \\
$2\times 10^{-3}$ & 0.0528   & 1.0604   & 0.2489 \\
$10^{-3}$         & 0.0142   & 1.0149   & 0.4908 \\
$5\times 10^{-4}$ & 0.0018   & 1.0018   & 0.7354 \\
$2\times 10^{-4}$ & 0.0001   & 1.0001   & 0.9473 \\
$10^{-4}$         & 0.000025 & 1.000025 & 0.9933 \\
$5\times 10^{-5}$ & 0.0000   & 1.0000   & 0.9994 \\
$2\times 10^{-5}$ & 0.0000   & 1.0000   & 1.0000 \\
\bottomrule
\end{tabular}
\end{table}

\begin{figure}[!t]
    \centering
    \includegraphics[width=\columnwidth,height=0.28\textheight,keepaspectratio]{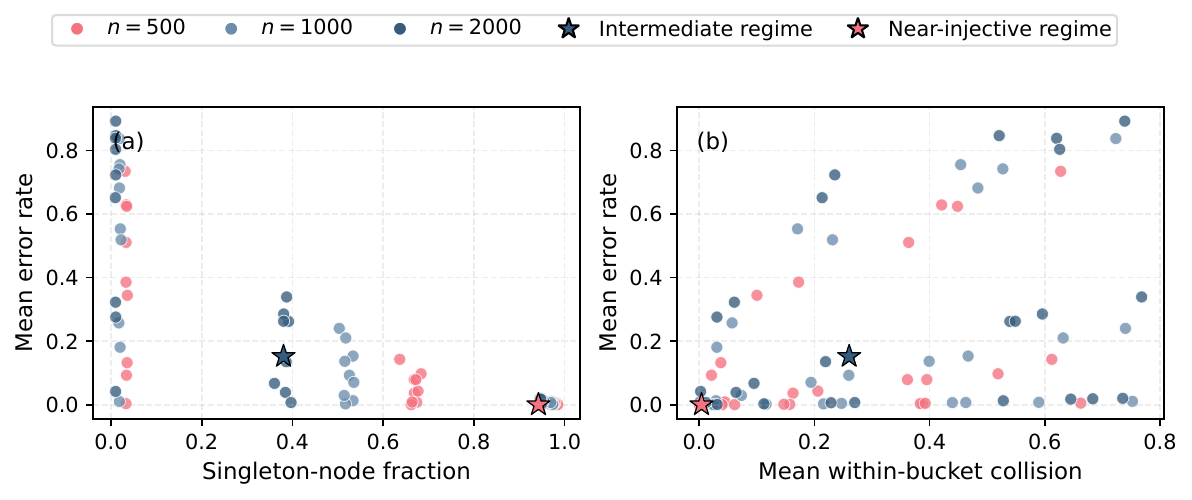}
    \caption{Bucketwise diagnostics on random $3$-regular graphs. Left: mean error vs.\ singleton-node fraction. Right: mean error vs.\ mean within-bucket collision. Highlighted points show an intermediate regime and a near-injective regime.}
    \label{fig:app_bucketwise}
\end{figure}

\end{document}